\documentclass{article}
 \usepackage{amssymb,nicefrac}
\usepackage{amsfonts}
\usepackage{graphicx}
\usepackage[fleqn]{amsmath}
\usepackage[varg]{txfonts}
\usepackage{stmaryrd}

\usepackage{framed}
\usepackage{color}
\usepackage[normalem]{ulem}
\usepackage{tikzfig}
\usepackage{placeins}
\usepackage{blkarray}
\usepackage{cleveref}

\usepackage{mathtools} 
\DeclarePairedDelimiter\bra{\langle}{\rvert}
\DeclarePairedDelimiter\ket{\lvert}{\rangle}
\newenvironment{proof}{\textbf{Proof:}}{\hfill$\Box$\newline}
\tikzstyle{env}=[copoint,regular polygon rotate=0,minimum width=0.2cm, fill=black]

\tikzstyle{every picture}=[baseline=-0.25em]
\tikzstyle{dotpic}=[scale=0.5]
\tikzstyle{diredges}=[every to/.style={diredge}]
\tikzstyle{dot graph}=[shorten <=-0.1mm,shorten >=-0.1mm,scale=0.6]
\tikzstyle{plot point}=[circle,fill=black,minimum width=2mm,inner sep=0]

\tikzstyle{braceedge}=[decorate,decoration={brace,amplitude=2mm,raise=-1mm}]
\tikzstyle{small braceedge}=[decorate,decoration={brace,amplitude=1mm,raise=-1mm}]
\tikzstyle{left hook arrow}=[left hook-latex]
\tikzstyle{right hook arrow}=[right hook-latex]

\tikzstyle{dtriangle}=[fill=yellow,draw=black,shape=isosceles triangle,shape border rotate=-90,isosceles triangle stretches=true,inner sep=0.8pt,minimum width=0.25cm,minimum height=2mm]
\tikzstyle{vtriang}=[fill=yellow,draw=black,shape=isosceles triangle,shape border rotate=180,isosceles triangle stretches=true,inner sep=0.8pt,minimum width=0.25cm,minimum height=2mm]
\tikzstyle{trigmc}=[fill=green,draw=black,shape=isosceles triangle,shape border rotate=90,isosceles triangle stretches=true,inner sep=0.8pt,minimum width=0.3cm,minimum height=2mm]
\tikzstyle{vrt}=[fill=yellow,draw=black,shape=isosceles triangle,shape border rotate=0,isosceles triangle stretches=true,inner sep=0.8pt,minimum width=0.25cm,minimum height=2mm]
\tikzstyle{H box}=[rectangle,fill=yellow,draw=black,xscale=0.8,yscale=0.8, inner sep=0.6pt]
\tikzstyle{gbox}=[rectangle,fill=green,draw=black,xscale=1.0,yscale=1.0, inner sep=1.pt]
\tikzstyle{rbox}=[rectangle,fill=red,draw=black,xscale=1.0,yscale=1.0, inner sep=1.pt]
\tikzstyle{zhbx}=[rectangle,fill=white,draw=black,xscale=1.0,yscale=1.0, inner sep=1.6pt]
\tikzstyle{newh}=[rectangle,fill=yellow,draw=black,xscale=2.0,yscale=2.0, inner sep=1.6pt]
\tikzstyle{dashededge}=[-, dashed] 
\tikzstyle{triangle}=[fill=yellow,draw=black,shape=isosceles triangle,shape border rotate=90,isosceles triangle stretches=true,inner sep=0.8pt,minimum width=0.25cm,minimum height=2mm]
\tikzstyle{bn}=[circle,fill=black,draw=black,scale=.4]
\tikzstyle{wn}=[circle,fill=white,draw=black,scale=.6]
\tikzstyle{dn}=[circle,fill=none,draw=gray]
\tikzstyle{bspider}=[fill=black,draw=black,scale=1,shape=isosceles triangle,shape border rotate=-90,isosceles triangle stretches=true,inner sep=1pt,minimum width=0.4cm,minimum height=3mm]
\tikzstyle{dbspider}=[fill=black,draw=black,scale=1,shape=isosceles triangle,shape border rotate=90,isosceles triangle stretches=true,inner sep=1pt,minimum width=0.4cm,minimum height=3mm]
\tikzstyle{L}=[rectangle,shape=rectangle,fill=green,draw=black]

\tikzstyle{Z dot}=[inner sep=0mm, minimum size=2mm, shape=circle, draw=black, fill={rgb,255: red,221; green,255; blue,221}]
\tikzstyle{Z phase dot}=[minimum size=5mm, font={\footnotesize\boldmath}, shape=rectangle, rounded corners=2mm, inner sep=0.2mm, outer sep=-2mm, scale=0.8, draw=black, fill={rgb,255: red,221; green,255; blue,221}]
\tikzstyle{X dot}=[Z dot, shape=circle, draw=black, fill={rgb,255: red,255; green,136; blue,136}]
\tikzstyle{X phase dot}=[Z phase dot, fill={rgb,255: red,255; green,136; blue,136}, font={\footnotesize\boldmath}]
\tikzstyle{hadamard edge}=[-, dashed, dash pattern=on 2pt off 0.5pt, thick, draw={rgb,255: red,68; green,136; blue,255}]
\tikzstyle{black dot}=[inner sep=0.7mm,minimum width=0pt,minimum height=0pt,fill=black,draw=black,shape=circle]

\tikzstyle{dot}=[black dot]
\tikzstyle{smalldot}=[inner sep=0.4mm,minimum width=0pt,minimum height=0pt,fill=black,draw=black,shape=circle]%NEW
\tikzstyle{white dot}=[dot,fill=white]
\tikzstyle{antipode}=[white dot,inner sep=0.3mm,font=\footnotesize]
\tikzstyle{smallwhitedot}=[smalldot,fill=white]%NEW
\tikzstyle{alt white dot}=[white dot,label={[xshift=3.07mm,yshift=-0.05mm,font=\footnotesize]left:$*$}]
\tikzstyle{gray dot}=[dot,fill=gray!40!white]
\tikzstyle{smallgraydot}=[smalldot,fill=gray!40!white]%NEW
\tikzstyle{box vertex}=[draw=black,rectangle]
\tikzstyle{small box}=[box vertex,fill=white]%% added rwd]
\tikzstyle{whitebg}=[fill=white,inner sep=2pt]
\tikzstyle{graph state vertex}=[sg vertex,fill=black]

\tikzstyle{wide copoint}=[fill=white,draw=black,shape=isosceles triangle,shape border rotate=90,isosceles triangle stretches=true,inner sep=1pt,minimum width=1.5cm,minimum height=5mm]
\tikzstyle{wide point}=[fill=white,draw=black,shape=isosceles triangle,shape border rotate=-90,isosceles triangle stretches=true,inner sep=1pt,minimum width=1.5cm,minimum height=4mm]
\tikzstyle{very wide copoint}=[fill=white,draw=black,shape=isosceles triangle,shape border rotate=-90,isosceles triangle stretches=true,inner sep=1pt,minimum width=2.5cm,minimum height=4mm]
\tikzstyle{very wide empty copoint}=[draw=black,shape=isosceles triangle,shape border rotate=-90,isosceles triangle stretches=true,inner sep=1pt,minimum width=2.5cm,minimum height=4mm]
\tikzstyle{symm}=[ultra thick,shorten <=-1mm,shorten >=-1mm]

\tikzstyle{square box}=[rectangle,fill=white,draw=black,minimum height=5mm,minimum width=5mm,font=\small]
\tikzstyle{square gray box}=[rectangle,fill=gray!30,draw=black,minimum height=6mm,minimum width=6mm]
\tikzstyle{copoint}=[regular polygon,regular polygon sides=3,draw=black,scale=0.75,inner sep=-0.5pt,minimum width=7mm,fill=white]
\tikzstyle{point}=[regular polygon,regular polygon sides=3,draw=black,scale=0.75,inner sep=-0.5pt,minimum width=7mm,fill=white,regular polygon rotate=180]
\tikzstyle{gray point}=[point,fill=gray!40!white]
\tikzstyle{gray copoint}=[copoint,fill=gray!40!white]

\newcommand{\edgearrow}{{\arrow[black]{>}}}
\newcommand{\edgetick}{{\arrow[black,scale=0.7,very thick]{|}}}
\tikzstyle{diredge}=[->]
\tikzstyle{rdiredge}=[<-]
\tikzstyle{medium diredge}=[->]

\tikzstyle{short diredge}=[->]
\tikzstyle{halfedge}=[-)]
\tikzstyle{other halfedge}=[(-]
\tikzstyle{freeedge}=[(-)]
\tikzstyle{white edge}=[line width=5pt,white]
\tikzstyle{tick}=[postaction=decorate,decoration={markings, mark=at position 0.5 with \edgetick}]
\tikzstyle{small map edge}=[|-latex, gray!60!blue, shorten <=0.9mm, shorten >=0.5mm]
\tikzstyle{thick dashed edge}=[very thick,dashed,gray!40]
\tikzstyle{map edge}=[|-latex,very thick, gray!40, shorten <=1mm, shorten >=0.5mm]
\tikzstyle{tickedge}=[postaction=decorate,
  decoration={markings, mark=at position 0.5 with \edgetick}]
\tikzstyle{dirtickedge}=[postaction=decorate,
  decoration={markings, mark=at position 0.5 with \edgetick},
  decoration={markings, mark=at position 0.85 with \edgearrow}]
\tikzstyle{dirdoubletickedge}=[postaction=decorate,
  decoration={markings, mark=at position 0.4 with \edgetick},
  decoration={markings, mark=at position 0.6 with \edgetick},
  decoration={markings, mark=at position 0.85 with \edgearrow}]

\makeatletter
\newcommand{\boxshape}[3]{%
\pgfdeclareshape{#1}{
\inheritsavedanchors[from=rectangle] % this is nearly a rectangle
\inheritanchorborder[from=rectangle]
\inheritanchor[from=rectangle]{center}
\inheritanchor[from=rectangle]{north}
\inheritanchor[from=rectangle]{south}
\inheritanchor[from=rectangle]{west}
\inheritanchor[from=rectangle]{east}
\backgroundpath{% this is new
\southwest \pgf@xa=\pgf@x \pgf@ya=\pgf@y
\northeast \pgf@xb=\pgf@x \pgf@yb=\pgf@y

\@tempdima=#2
\@tempdimb=#3

\pgfpathmoveto{\pgfpoint{\pgf@xa - 5pt + \@tempdima}{\pgf@ya}}
\pgfpathlineto{\pgfpoint{\pgf@xa - 5pt - \@tempdima}{\pgf@yb}}
\pgfpathlineto{\pgfpoint{\pgf@xb + 5pt + \@tempdimb}{\pgf@yb}}
\pgfpathlineto{\pgfpoint{\pgf@xb + 5pt - \@tempdimb}{\pgf@ya}}
\pgfpathlineto{\pgfpoint{\pgf@xa - 5pt + \@tempdima}{\pgf@ya}}
\pgfpathclose
}
}}

\boxshape{NEbox}{0pt}{8pt}
\boxshape{SEbox}{0pt}{-8pt}
\boxshape{NWbox}{8pt}{0pt}
\boxshape{SWbox}{-8pt}{0pt}
\makeatother

\tikzstyle{map}=[draw,shape=NEbox,inner sep=7pt]
\tikzstyle{mapdag}=[draw,shape=SEbox,inner sep=7pt]
\tikzstyle{maptrans}=[draw,shape=SWbox,inner sep=7pt]
\tikzstyle{mapconj}=[draw,shape=NWbox,inner sep=7pt]

\tikzstyle{probs}=[shape=semicircle,fill=gray!40!white,draw=black,shape border rotate=180,minimum width=1.2cm]

\tikzstyle{arrs}=[-latex,font=\small,auto]
\tikzstyle{arrow plain}=[arrs]
\tikzstyle{arrow dashed}=[dashed,arrs]
\tikzstyle{arrow bold}=[very thick,arrs]
\tikzstyle{arrow hide}=[draw=white!0,-]
\tikzstyle{arrow reverse}=[latex-]
\tikzstyle{cdnode}=[]

\tikzstyle{gn}=[dot,fill=green,minimum width=0.25cm,inner sep=0pt]
\tikzstyle{rno}=[dot,fill=red,inner sep=0pt,minimum width=0.25cm]
\tikzstyle{rn}=[dot,fill=pink,inner sep=0pt,minimum width=0.25cm]

\tikzstyle{rc}=[dot,thick,fill=white,draw = red,minimum width=0.3cm,inner sep=0pt]
\tikzstyle{gc}=[dot,thick,fill=white,draw= green,inner sep=0pt,minimum width=0.3cm]
\tikzstyle{bc}=[dot,thick,fill=white,draw= blue,minimum width=0.3cm]

\tikzstyle{label}=[circle,fill=white,minimum width=0.3cm]

\tikzstyle{clocklabel}=[dot,fill=yellow,draw=black,font=\tiny,inner sep=0.75pt]

\tikzstyle{rsn}=[circle split,draw,fill=red,font=\tiny,inner sep=0.75pt]
\tikzstyle{gsn}=[circle split,draw,fill=green,font=\tiny,inner sep=0.75pt]
\tikzstyle{bsn}=[circle split,draw,fill=blue,font=\tiny,inner sep=0.75pt]

\tikzstyle{rsc}=[circle split,thick,draw= red,draw,fill=white,font=\tiny,inner sep=0.75pt]
\tikzstyle{gsc}=[circle split,thick,draw= green,draw,fill=white,font=\tiny,inner sep=0.75pt]
\tikzstyle{bsc}=[circle split,thick,draw= blue,draw,fill=white,font=\tiny,inner sep=0.75pt]

\tikzstyle{cnot}=[fill=white,shape=circle,inner sep=-1.4pt]
\tikzstyle{wire label}=[font=\tiny, auto]

\tikzstyle{cdiag}=[matrix of math nodes, row sep=3em, column sep=3em, text height=1.5ex, text depth=0.25ex,inner sep=0.5em]
\tikzstyle{arrow above}=[transform canvas={yshift=0.5ex}]
\tikzstyle{arrow below}=[transform canvas={yshift=-0.5ex}]

\newtheorem{Th}{Theorem}[section]
\newtheorem{theorem}[Th]{Theorem}
 
\newtheorem{lemma}[Th]{Lemma}
\newtheorem{corollary}[Th]{Corollary}

\newtheorem{remark}[Th]{Remark}

\makeatletter
\newcommand{\vast}{\bBigg@{6.5}}
\makeatother

\newcommand{\vertrule}[1][1ex]{\rule{.4pt}{#1}}

\newcommand{\rddots}{\rotatebox{90}{$\ddots$}}%====right diagonal dots 
\title{ Qufinite ZX-calculus: a unified framework of qudit ZX-calculi }
\author{Quanlong Wang\\
Cambridge Quantum/ Quantinuum }
\begin{document}
%\maketitle
\date{}\maketitle
\begin{abstract}
ZX-calculus is graphical language for quantum computing which usually focuses on qubits. In this paper, we generalise qubit ZX-calculus to qudit ZX-calculus in any finite dimension by introducing suitable generators, especially a carefully chosen triangle node. As a consequence we  obtain a set of rewriting rules which can be seen as a direct generalisation of qubit rules, and a normal form for any qudit vectors. Based on the qudit ZX-calculi, we propose a graphical formalism called qufinite ZX-calculus as a unified framework for all qudit ZX-calculi, which is universal for finite quantum theory due to a normal form for matrix of any finite size. As a result, it would be interesting to give a fine-grained version of  the diagrammatic reconstruction of finite quantum theory \cite{Selby2021reconstructing} within the framework of qufinite ZX-calculus.
%We would expect a reconstruction of finite quantum theory within the framework of qufinite ZX-calculus which focuses on compositionality  without resorting to any probability theory or sum structures.

\end{abstract}

%===============================================================

  \section{ Introduction}
 %\tikzfig{TikZit/syw}  % Intuitively but also reasonably describe ideas about this paper. Make a good story!

 ZX-calculus is invented as a graphical language for quantum computing \cite{CoeckeDuncan}, so no wonder it is overwhelmingly concentrated on qubits, i.e., each of its wire in a diagram represents a 2-dimensional system and each diagram corresponds to a matrix of size $2^m\times 2^n$. Mathematically speaking, ZX-calculus is based on a compact closed PROP \cite{BaezCR} represented by Z spiders and X spiders (depicted in green and red respectively in this paper) as main generators as well as  rewriting rules  which are equalities of diagrams composed vertically or horizontally in terms of those generators. The ZX-calculus is called universal if each matrix of $2^m\times 2^n$ can be represented by a ZX diagram \cite{CoeckeDuncan}, and is called complete if any two diagrams corresponding to the same matrix can be rewritten into each other with ZX rules \cite{amarngwanglics, jpvbeyondlics}.

However, qubits are not the whole story. In fact, many  quantum systems naturally exhibit more than two dimensions \cite{Farinholt_2014}, so it would be very useful to have  ZX-calculus for higher dimensional quantum systems, i.e., qudits of dimension $d\geq 2$. Actually, ZX-calculus has been partly generalised to any finite dimension in \cite{Ranchin},  and was even more proved to be universal in \cite{BianWang2}. But that proof of universality is not a constructive one, which means given a matrix it is not easy in general to construct its corresponding diagram, although in principle this can be done. Furthermore, from the preliminary version of qudit ZX-calculus it is non-trivial to find a way that could lead to a proof of completeness. 

On the other hand, qubit ZX-calculus has been proved to be complete via a translation from another graphical calculus called ZW-calculus which is already complete \cite{amarngwanglics, jpvbeyondlics}. Although there exists a qudit version of ZW-calculus which is universal due to a normal form \cite{amar},  it is unclear how completeness for qudit  ZW-calculus could be achieved,  thus makes it unfeasible to obtain completeness of qudit ZX-calculus by the translation method. 
 
 Furthermore, each wire of a qudit ZX diagram represents a $d$-dimensional system, while in finite dimensional quantum theory it is quite common that the involved systems have hybrid dimensions instead of some power of $d$. Therefore, it is desirable to have a unified framework for qudit ZX-calculi in all finite dimensions.   
  
  In this paper, we first generalise qubit ZX-calculus to qudit ZX-calculus in any dimension $d\geq 2$. The generators we give here are different from the previous ones \cite{Ranchin} in that the Z spider has a phase vector composed of complex numbers rather than real numbers,  the generalised Hadamard node and its  adjoint are unnormalised, and a generalised triangle node (in comparing to the one given in \cite{jpvcltlics, amarngwanglics}) and its inverse are added in.  As a consequence, the X spider is also  unnormalised, and most of the qubit rewriting rules as given in \cite{qwangnormalformbit} are generalised to the qudit case for any dimension. Note that some of these qudit rewriting rules have been presented in \cite{qwangslides19} and \cite{cwiconscious}. Most importantly,  the normal for qubit vectors  \cite{qwangnormalformbit}  has been generalised for qudit vectors, which means universality of  qudit ZX-calculus and a feasible approach for proof of completeness (similar to the method for qubit ZX shown in \cite{qwangnormalformbit}).

Then we propose a formalism called qufinite ZX-calculus as a unified framework for qudit ZX-calculi in all finite dimensions.  The key idea is to label each wire with its dimension and add two new generators called dimension-splitter and dimension-binder respectively. This formalism is not a PROP anymore, but still a compact closed category.  With the new generators and a normal form for qudits, we construct a normal form for arbitrary matrix of size $m\times n$. Therefore, the qufinite ZX-calculus is universal for finite dimensional quantum theory.     

 Finally we mention that the generators  (except for the Hadamard node and its adjoint) and the normal form of  the qufinite ZX-calculus can be generalised to be over arbitrary commutative semirings. Hence we can have a qufinite ZX-calculus over commutative semirings which is universal and possibly can be proved to be complete following the method given in  \cite{qwangrsmring}.
 
 % Why do we need qudit ZX-calculus and qufinite ZX-calculus? Quantum computing Mention the symbol in SZX, which is similar to dimension splitting box in qubit case.(once we have  a normal form, then to prove completeness means you need only prove that two diagrams which corresponding to the same matrix can be rewritten into the same normal form.) In this paper, we generalise algebraic qubit ZX-calculus to qudit ZX-calculus such that we can have a normal form for qudits. We also indicate a way to achieve completeness. 
  
%  There is is a qudit ZW-calculus with a normal form given in Arma's thesis, but no way to obtain its completeness. 
  %\textcolor{red}{To represent qudit triangle in terms of red and green spiders, we need qudit vertions of Toffoli gate, both in triangle form and spider form. To get the spider form, we need to know the phase gadget  decomposition of qudit Toffoli gate, this could be obtained by generalising qubit case in rewriting from phase decomposition of Toffoli gate to its triangle form. Maybe CCZ gate is enough, so H gate not needed.  }
  
%  \iffalse
   \section{ Generators and rules of qudit ZX-calculus}
   In this paper, $d$ is an integer and $d\geq 2$, $\mathbb C$ is the filed of complex numbers. All the diagrams are read from top to bottom. Similar to the qubit case, qudit ZX-calculus is based on a PROP which can be represented by generators and rewriting rules. 
   
   First we give the generators of qudit ZX-calculus. 
 \begin{table}
\begin{center} % \tikzfig{TikZit/emptysquare-small}
\begin{tabular}{|r@{~}r@{~}c@{~}c|r@{~}r@{~}c@{~}c|}
\hline
%$R_{Z,\alpha}^{(n,m)}$&$:$&$n\to m$ & \tikzfig{TikZit/generalgreenspiderqdit2}  & $R_{X,\alpha}^{(n,m)}$&$:$&$n\to m$ & \tikzfig{TikZit/quditrspider}\\ \hline

$R_{Z,\overrightarrow{a}}^{(n,m)}$&$:$&$n\to m$ & %
	\beginpgfgraphicnamed{TikZit/generalgreenspiderqdit2}
	\InputIfFileExists{TikZit/generalgreenspiderqdit2.tikz}{}{\input{./figures/TikZit/generalgreenspiderqdit2.tikz}}%
	\endpgfgraphicnamed
  & & & &\\ \hline

%$R_{Z,\alpha}^{(n,m)}$&$:$&$n\to m$ & \tikzfig{TikZit/quditgspider}  & $R_{X,\alpha}^{(n,m)}$&$:$&$n\to m$ & \tikzfig{TikZit/quditrspider}\\
%\multicolumn{3}{|c}{ \tikzfig{TikZit/generator_spider2} }&(P1)\\ \hline
$H$&$:$&$1\to 1$ &%
	\beginpgfgraphicnamed{TikZit/HadaDecomSingleslt}
	\InputIfFileExists{TikZit/HadaDecomSingleslt.tikz}{}{\input{./figures/TikZit/HadaDecomSingleslt.tikz}}%
	\endpgfgraphicnamed

 &  $H^{\dagger}$&$:$&$ 1\to 1$& %
	\beginpgfgraphicnamed{TikZit/RGg_Hadad}
	\InputIfFileExists{TikZit/RGg_Hadad.tikz}{}{\input{./figures/TikZit/RGg_Hadad.tikz}}%
	\endpgfgraphicnamed
\\\hline
   $\mathbb I$&$:$&$1\to 1$&%
	\beginpgfgraphicnamed{TikZit/Id}
	\InputIfFileExists{TikZit/Id.tikz}{}{\input{./figures/TikZit/Id.tikz}}%
	\endpgfgraphicnamed
 &   $\sigma$&$:$&$ 2\to 2$& %
	\beginpgfgraphicnamed{TikZit/swap}
	\InputIfFileExists{TikZit/swap.tikz}{}{\input{./figures/TikZit/swap.tikz}}%
	\endpgfgraphicnamed
\\ \hline
   $C_a$&$:$&$ 0\to 2$& %
	\beginpgfgraphicnamed{TikZit/cap}
	\InputIfFileExists{TikZit/cap.tikz}{}{\input{./figures/TikZit/cap.tikz}}%
	\endpgfgraphicnamed
 &$ C_u$&$:$&$ 2\to 0$&%
	\beginpgfgraphicnamed{TikZit/cup}
	\InputIfFileExists{TikZit/cup.tikz}{}{\input{./figures/TikZit/cup.tikz}}%
	\endpgfgraphicnamed
 \\\hline
 $T_g$&$:$&$1\to 1$&%
	\beginpgfgraphicnamed{TikZit/triangle}
	\InputIfFileExists{TikZit/triangle.tikz}{}{\input{./figures/TikZit/triangle.tikz}}%
	\endpgfgraphicnamed
  & $T_g^{-1}$&$:$&$1\to 1$&%
	\beginpgfgraphicnamed{TikZit/triangleinv}
	\InputIfFileExists{TikZit/triangleinv.tikz}{}{\input{./figures/TikZit/triangleinv.tikz}}%
	\endpgfgraphicnamed
 \\\hline
\end{tabular}\caption{ Generators of qudit ZX-calculus, where $m,n\in \mathbb N$, $\protect\overrightarrow{a}=(a_1, \cdots, a_{d-1}), a_i \in \mathbb C. $} \label{qbzxgeneratordit}
\end{center}
\end{table}
\FloatBarrier

 For convenience, we make the following denotations: 
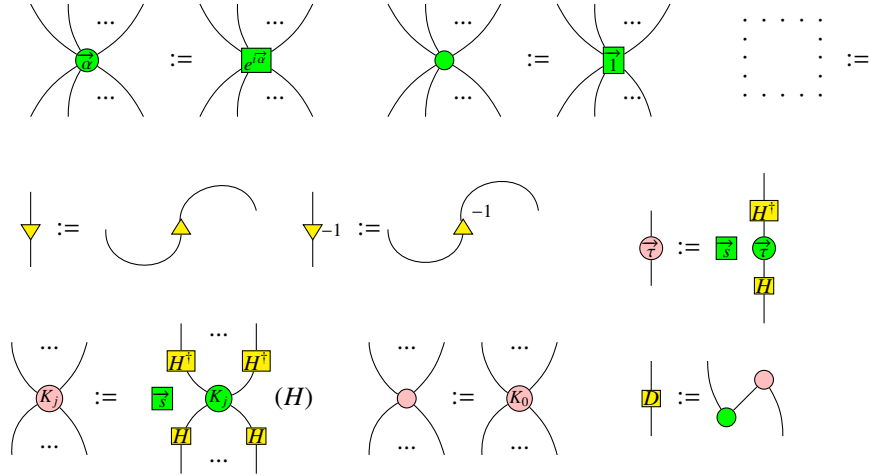
\[
% \tikzfig{TikZit/spidersdenotedit} 
 %
	\beginpgfgraphicnamed{TikZit/spidersdenotedit3}
	\InputIfFileExists{TikZit/spidersdenotedit3.tikz}{}{\input{./figures/TikZit/spidersdenotedit3.tikz}}%
	\endpgfgraphicnamed
 
\]
 where $   \overrightarrow{\alpha}=(\alpha_1, \cdots, \alpha_{d-1}),  \alpha_i \in \mathbb R,  e^{i\overrightarrow{\alpha}}=(e^{i\alpha_1}, \cdots, e^{i\alpha_{d-1}}), \overrightarrow{1}=\overbrace{(1,\cdots,1)}^{d-1}, \overrightarrow{s}=(0,\cdots,0, \frac{1}{d}-1), \overrightarrow{\tau}=(\tau_1, \cdots, \tau_k, \cdots, \tau_{d-1}), \tau_k=k\pi+\frac{k^2\pi}{d}, 1\leq k \leq d-1,  K_j=(j\frac{2\pi}{d}, 2j\frac{2\pi}{d}, \cdots, (d-1)j\frac{2\pi}{d}),  0\leq j \leq d-1$.
 
 The diagrams of the qudit ZX-calculus have the following standard interpretation $ \left\llbracket \cdot \right\rrbracket$:

 \[
 \left\llbracket %
	\beginpgfgraphicnamed{TikZit/generalgreenspiderqdit2}
	\InputIfFileExists{TikZit/generalgreenspiderqdit2.tikz}{}{\input{./figures/TikZit/generalgreenspiderqdit2.tikz}}%
	\endpgfgraphicnamed
 \right\rrbracket=\sum_{j=0}^{d-1}a_j\ket{j}^{\otimes m}\bra{j}^{\otimes n}, a_0=1,  \protect\overrightarrow{a}=(a_1, \cdots, a_{d-1}), a_i \in \mathbb C. 
\]

 \[
\left\llbracket %
	\beginpgfgraphicnamed{TikZit/quditrspiderclassic}
	\InputIfFileExists{TikZit/quditrspiderclassic.tikz}{}{\input{./figures/TikZit/quditrspiderclassic.tikz}}%
	\endpgfgraphicnamed
 \right\rrbracket=\sum_{\substack{0\leq i_1, \cdots, i_m,  j_1, \cdots, j_n\leq d-1\\ i_1+\cdots+ i_m+j\equiv  j_1+\cdots +j_n(mod~ d)}}\ket{i_1, \cdots, i_m}\bra{j_1, \cdots, j_n}, \quad K_j=(j\frac{2\pi}{d}, 2j\frac{2\pi}{d}, \cdots, (d-1)j\frac{2\pi}{d}),  0\leq j \leq d-1,
\]

\[
\left\llbracket%
	\beginpgfgraphicnamed{TikZit/HadaDecomSingleslt}
	\begin{tikzpicture}
	\begin{pgfonlayer}{nodelayer}
		\node [style=H box] (0) at (-0.75, 0) {$H$};
		\node [style=none] (1) at (-0.75, -0.5) {};
		\node [style=none] (2) at (-0.75, 0.5) {};
	\end{pgfonlayer}
	\begin{pgfonlayer}{edgelayer}
		\draw (2.center) to (0);
		\draw (1.center) to (0);
	\end{pgfonlayer}
\end{tikzpicture}}%
	\endpgfgraphicnamed
\right\rrbracket=\sum_{k, j=0}^{d-1}\xi^{jk}\ket{j}\bra{k},  \xi=e^{i\frac{2\pi}{d}}, \quad
\left\llbracket%
	\beginpgfgraphicnamed{TikZit/RGg_Hadad}
	\begin{tikzpicture}
	\begin{pgfonlayer}{nodelayer}
		\node [style={H box}] (0) at (0, -0) {$H^\dagger$};
		\node [style=none] (1) at (0, -0.5) {};
		\node [style=none] (2) at (0, 0.5) {};
	\end{pgfonlayer}
	\begin{pgfonlayer}{edgelayer}
		\draw (2.center) to (0);
		\draw (1.center) to (0);
	\end{pgfonlayer}
\end{tikzpicture}}%
	\endpgfgraphicnamed
\right\rrbracket=\sum_{k, j=0}^{d-1}\bar{\xi}^{jk}\ket{j}\bra{k}, \bar{\xi}=e^{-i\frac{2\pi}{d}},  \]
\[
  \left\llbracket%
	\beginpgfgraphicnamed{TikZit/triangle}
	\begin{tikzpicture}
	\begin{pgfonlayer}{nodelayer}
		\node [style=none] (0) at (0, 0.5) {};
		\node [style=triangle] (1) at (0, 0) {};
		\node [style=none] (2) at (0, -0.5) {};
	\end{pgfonlayer}
	\begin{pgfonlayer}{edgelayer}
		\draw (0.center) to (2.center);
	\end{pgfonlayer}
\end{tikzpicture}}%
	\endpgfgraphicnamed
\right\rrbracket=I_d+\sum_{i=1}^{d-1}\ket{0}\bra{i}, \quad
  %\ket{0}\bra{0}+\sum_{i=1}^{d-1}(\ket{0}+\ket{i})\bra{i}, \quad
   \left\llbracket%
	\beginpgfgraphicnamed{TikZit/triangleinv}
	\begin{tikzpicture}
	\begin{pgfonlayer}{nodelayer}
		\node [style=none] (0) at (0.25, 0.25) {-{\scriptsize1}};
		\node [style=triangle] (1) at (0, 0) {};
		\node [style=none] (2) at (0, -0.5) {};
		\node [style=none] (3) at (0, 0.5) {};
	\end{pgfonlayer}
	\begin{pgfonlayer}{edgelayer}
		\draw (3.center) to (2.center);
	\end{pgfonlayer}
\end{tikzpicture}}%
	\endpgfgraphicnamed
\right\rrbracket=I_d-\sum_{i=1}^{d-1}\ket{0}\bra{i},
   \]
   \[
      \left\llbracket%
	\beginpgfgraphicnamed{TikZit/redtaugate}
	\begin{tikzpicture}
	\begin{pgfonlayer}{nodelayer}
		\node [style=none] (0) at (0, -0.5) {};
		\node [style=none] (1) at (0, 0.5) {};
		\node [style=rn] (2) at (0, 0) { ${\scriptstyle  \overrightarrow{\tau}}$ };
	\end{pgfonlayer}
	\begin{pgfonlayer}{edgelayer}
		\draw (2) to (1.center);
		\draw (2) to (0.center);
	\end{pgfonlayer}
\end{tikzpicture}}%
	\endpgfgraphicnamed
\right\rrbracket=\frac{1}{d}\sum_{k,l,n=0}^{d-1}e^{i\tau_l}\xi^{(k-n)l}\ket{k}\bra{n}, \overrightarrow{\tau}=(\tau_1, \cdots, \tau_k, \cdots, \tau_{d-1}), \tau_k=k\pi+\frac{k^2\pi}{d}, 0\leq k \leq d-1,
   \]

\[
   \left\llbracket%
	\beginpgfgraphicnamed{TikZit/Id}
	\begin{tikzpicture}
	\begin{pgfonlayer}{nodelayer}
		\node [style=none] (1) at (0.5, 0.3) {};
		\node [style=none] (2) at (0.5, -0.3) {};
		\node [style=none] (3) at (0.5, -0.5) {};
		\node [style=none] (4) at (0.5, 0.5) {};
	\end{pgfonlayer}
	\begin{pgfonlayer}{edgelayer}
		\draw (1.center) to (2.center);
	\end{pgfonlayer}
\end{tikzpicture}}%
	\endpgfgraphicnamed
\right\rrbracket=I_d=\sum_{j=0}^{d-1}\ket{j}\bra{j},  \quad
 \left\llbracket%
	\beginpgfgraphicnamed{TikZit/swap}
	\InputIfFileExists{TikZit/swap.tikz}{}{\input{./figures/TikZit/swap.tikz}}%
	\endpgfgraphicnamed
\right\rrbracket=\sum_{i, j=0}^{d-1}\ket{ji}\bra{ij},\quad
  \left\llbracket%
	\beginpgfgraphicnamed{TikZit/cap}
	\begin{tikzpicture}
	\begin{pgfonlayer}{nodelayer}
		\node [style=none] (0) at (-0.5, -0.25) {};
		\node [style=none] (1) at (0.5, -0.25) {};
	\end{pgfonlayer}
	\begin{pgfonlayer}{edgelayer}
		\draw [bend left=90, looseness=1.50] (0.center) to (1.center);
	\end{pgfonlayer}
\end{tikzpicture}}%
	\endpgfgraphicnamed
\right\rrbracket=\sum_{j=0}^{d-1}\ket{jj}, \quad
   \left\llbracket%
	\beginpgfgraphicnamed{TikZit/cup}
	\begin{tikzpicture}
	\begin{pgfonlayer}{nodelayer}
		\node [style=none] (0) at (-0.5, 0.25) {};
		\node [style=none] (1) at (0.5, 0.25) {};
	\end{pgfonlayer}
	\begin{pgfonlayer}{edgelayer}
		\draw [bend right=90, looseness=1.50] (0.center) to (1.center);
	\end{pgfonlayer}
\end{tikzpicture}}%
	\endpgfgraphicnamed
\right\rrbracket=\sum_{j=0}^{d-1}\bra{jj},   \quad \left\llbracket%
	\beginpgfgraphicnamed{TikZit/emptysquare}
	\InputIfFileExists{TikZit/emptysquare.tikz}{}{\input{./figures/TikZit/emptysquare.tikz}}%
	\endpgfgraphicnamed
\right\rrbracket=1,
      \]

\[  \llbracket D_1\otimes D_2  \rrbracket =  \llbracket D_1  \rrbracket \otimes  \llbracket  D_2  \rrbracket, \quad 
 \llbracket D_1\circ D_2  \rrbracket =  \llbracket D_1  \rrbracket \circ  \llbracket  D_2  \rrbracket.
  \]
  
 In particular,
 \[
\left\llbracket %
	\beginpgfgraphicnamed{TikZit/quditclassicpoints}
	\begin{tikzpicture}
	\begin{pgfonlayer}{nodelayer}
		\node [style=none] (0) at (0, -0.25) {};
		\node [style=rn] (1) at (0, 0.25) {${\scriptstyle K_{j}}$};
	\end{pgfonlayer}
	\begin{pgfonlayer}{edgelayer}
		\draw (1) to (0.center);
	\end{pgfonlayer}
\end{tikzpicture}}%
	\endpgfgraphicnamed
 \right\rrbracket=\ket{d-j}, \quad\left\llbracket %
	\beginpgfgraphicnamed{TikZit/quditclassiccopoints}
	\begin{tikzpicture}
	\begin{pgfonlayer}{nodelayer}
		\node [style=rn] (0) at (0, -0.25) {${\scriptstyle K_{j}}$};
		\node [style=none] (1) at (0, 0.25) {};
	\end{pgfonlayer}
	\begin{pgfonlayer}{edgelayer}
		\draw (0) to (1.center);
	\end{pgfonlayer}
\end{tikzpicture}}%
	\endpgfgraphicnamed
 \right\rrbracket=\bra{j}, 0\leq j \leq d-1 (\ket{0}=\ket{d}).
\] 
%It can be calculated from (H) that
%\[
%\left\llbracket\tikzfig{TikZit/kjgate}\right\rrbracket=\sum_{i=0}^{d-1}\ket{i}\bra{i+j}, 0\leq j \leq d-1,\] where $+$ is a modulo $d$ operation.
 
   \begin{remark}
The interpretation of the $H$ node is the usual quantum Fourier transform without the scalar $\frac{1}{\sqrt{d}}$ \cite{Nielsen}, and the  $H^{\dagger}$ node is the unnormalised inverse quantum Fourier transform. Also note that $\tau_{d-k}\equiv \tau_k (mod 2\pi)$, i.e.,   $ \overrightarrow{\tau}=\overleftarrow{\tau} $.  
 
In the qubit case, the triangle node has the following interpretation:
 \[
  \left\llbracket%
	\beginpgfgraphicnamed{TikZit/triangle}
	\begin{tikzpicture}
	\begin{pgfonlayer}{nodelayer}
		\node [style=none] (0) at (0, 0.5) {};
		\node [style=triangle] (1) at (0, 0) {};
		\node [style=none] (2) at (0, -0.5) {};
	\end{pgfonlayer}
	\begin{pgfonlayer}{edgelayer}
		\draw (0.center) to (2.center);
	\end{pgfonlayer}
\end{tikzpicture}}%
	\endpgfgraphicnamed
\right\rrbracket=\begin{pmatrix}
        1 & 1 \\
        0 & 1
 \end{pmatrix}
  \]
When generalising the triangle node to higher dimensional cases, there could be multiple choices. Here we adopt the following form which turns out to be very useful:
  \[
    \left\llbracket%
	\beginpgfgraphicnamed{TikZit/triangle}
	\begin{tikzpicture}
	\begin{pgfonlayer}{nodelayer}
		\node [style=none] (0) at (0, 0.5) {};
		\node [style=triangle] (1) at (0, 0) {};
		\node [style=none] (2) at (0, -0.5) {};
	\end{pgfonlayer}
	\begin{pgfonlayer}{edgelayer}
		\draw (0.center) to (2.center);
	\end{pgfonlayer}
\end{tikzpicture}}%
	\endpgfgraphicnamed
\right\rrbracket=\begin{pmatrix}
      1 & 1 & \cdots & 1  \\
          & 1 &  \cdots & 0 \\
           &  & \ddots & \vdots  \\
         &  &  & 1 
 \end{pmatrix}
  \]
where the elements of the first row and the diagonal are $1$ and the entries in other places are  just $0$.  
 \end{remark}

 \begin{remark}
 As once can see from the standard interpretation of the diagrams, if we add in the red spiders  labelled with $K_j$  while drop off the $H$ node and the $H^{\dagger}$ node from the Table \ref{qbzxgeneratordit} of the qudit generators, then we get the generators of qudit ZX-calculus over arbitrary commutative semirings in any finite dimension, as have been shown in  \cite{cwiconscious}. 
  \end{remark}
  
  Below we show some rules of qudit ZX-calculus in any dimension $d\geq 2$ which are directly generalised from the rules for qubits  \cite{qwangnormalformbit}. 
   \begin{figure}[!h]
\begin{center} 
\[
\quad \qquad\begin{array}{|cccc|}
\hline
	\beginpgfgraphicnamed{TikZit/gengspiderfusedit}
	\InputIfFileExists{TikZit/gengspiderfusedit.tikz}{}{\input{./figures/TikZit/gengspiderfusedit.tikz}}%
	\endpgfgraphicnamed
&(S1) &%
	\beginpgfgraphicnamed{TikZit/s2qudit}
	\InputIfFileExists{TikZit/s2qudit.tikz}{}{\input{./figures/TikZit/s2qudit.tikz}}%
	\endpgfgraphicnamed
 &(S2)\\
	\beginpgfgraphicnamed{TikZit/s3qudit}
	\InputIfFileExists{TikZit/s3qudit.tikz}{}{\input{./figures/TikZit/s3qudit.tikz}}%
	\endpgfgraphicnamed
&(S3) & %
	\beginpgfgraphicnamed{TikZit/rdotaemptydit0}
	\InputIfFileExists{TikZit/rdotaemptydit0.tikz}{}{\input{./figures/TikZit/rdotaemptydit0.tikz}}%
	\endpgfgraphicnamed
  &(Ept) \\
%\tikzfig{TikZit/redspider0pfusedit2}&(S4) & &\\
  & &&\\ 
%\multicolumn{3}{|c}{\tikzfig{TikZit/redspider0pfusedit2}}&(S4)\\
  %& &&\\ 
  %
	\beginpgfgraphicnamed{TikZit/b1qudit}
	\InputIfFileExists{TikZit/b1qudit.tikz}{}{\input{./figures/TikZit/b1qudit.tikz}}%
	\endpgfgraphicnamed
&(B1)  & %
	\beginpgfgraphicnamed{TikZit/b2qudit}
	\InputIfFileExists{TikZit/b2qudit.tikz}{}{\input{./figures/TikZit/b2qudit.tikz}}%
	\endpgfgraphicnamed
&(B2)\\ 
    & &&\\ 
	\beginpgfgraphicnamed{TikZit/b3qudit}
	\InputIfFileExists{TikZit/b3qudit.tikz}{}{\input{./figures/TikZit/b3qudit.tikz}}%
	\endpgfgraphicnamed
  &(B3 )& %
	\beginpgfgraphicnamed{TikZit/pimultiplecpdit}
	\InputIfFileExists{TikZit/pimultiplecpdit.tikz}{}{\input{./figures/TikZit/pimultiplecpdit.tikz}}%
	\endpgfgraphicnamed
&(K1) \\
 & &&\\ 
	\beginpgfgraphicnamed{TikZit/k2adit}
	\InputIfFileExists{TikZit/k2adit.tikz}{}{\input{./figures/TikZit/k2adit.tikz}}%
	\endpgfgraphicnamed
  &(K2 )& %
	\beginpgfgraphicnamed{TikZit/euqdit}
	\InputIfFileExists{TikZit/euqdit.tikz}{}{\input{./figures/TikZit/euqdit.tikz}}%
	\endpgfgraphicnamed
 & (EU)\\
 & &&\\ 
	\beginpgfgraphicnamed{TikZit/zerotoreddit0}
	\InputIfFileExists{TikZit/zerotoreddit0.tikz}{}{\input{./figures/TikZit/zerotoreddit0.tikz}}%
	\endpgfgraphicnamed
&(Zer)&%
	\beginpgfgraphicnamed{TikZit/hhdaggerchangedit}
	\InputIfFileExists{TikZit/hhdaggerchangedit.tikz}{}{\input{./figures/TikZit/hhdaggerchangedit.tikz}}%
	\endpgfgraphicnamed
&(H1)\\
%\tikzfig{TikZit/k2ndit}&(K2)&\tikzfig{TikZit/h1}&(H1)\\
%\tikzfig{TikZit/h2dit}&(H2)&\tikzfig{TikZit/h2primedit}&(H2^\prime)\\
%\multicolumn{3}{|c}{\tikzfig{TikZit/p1sdit2}}&(P1)\\     & &&\\ 
%
	\beginpgfgraphicnamed{TikZit/p1sdit2}
	\InputIfFileExists{TikZit/p1sdit2.tikz}{}{\input{./figures/TikZit/p1sdit2.tikz}}%
	\endpgfgraphicnamed
&(P1)&%
	\beginpgfgraphicnamed{TikZit/dcomwtha0}
	\InputIfFileExists{TikZit/dcomwtha0.tikz}{}{\input{./figures/TikZit/dcomwtha0.tikz}}%
	\endpgfgraphicnamed
&(D1)\\
	\beginpgfgraphicnamed{TikZit/scalardit}
	\InputIfFileExists{TikZit/scalardit.tikz}{}{\input{./figures/TikZit/scalardit.tikz}}%
	\endpgfgraphicnamed
&(Sca) &  &\\ 
  		  		\hline  
  		\end{array}\]      
  	\end{center}
  		\caption{Qudit ZX-calculus rules I, where $\protect\overrightarrow{a}=(a_1,\cdots, a_{d-1}), \protect\overleftarrow{a}=(a_{d-1}, \cdots, a_1),  \protect\overrightarrow{a^{\prime}}=(\frac{a_{1-j}}{a_{d-j}},\cdots, \frac{a_{d-1-j}}{a_{d-j}}), \protect\overrightarrow{a_{-j}}=(0, \cdots, 0, a_{d-j}-1),  j \in \{0, 1,\cdots, d-1\}, a_0=a_d=1, \protect\overrightarrow{b}=(b_1,\cdots, b_{d-1}), \protect\overrightarrow{ab}=(a_1b_1,\cdots, a_{d-1}b_{d-1}), a_k, b_k\in \mathbb C, k \in \{1,\cdots, d-1\}, m  \in  \mathbb N, \protect\overrightarrow{\tau}=(\tau_1, \cdots, \tau_k, \cdots, \tau_{d-1}), \tau_k=k\pi+\frac{k^2\pi}{d}, 0\leq k \leq d-1; K_j=(j\frac{2\pi}{d}, 2j\frac{2\pi}{d}, \cdots, (d-1)j\frac{2\pi}{d}),  0\leq j \leq d-1, \protect\overrightarrow{0}=\protect\overbrace{(0,\cdots,0)}^{d-1},  \protect\overrightarrow{\sigma_i}=(\protect\underbrace{0,\cdots, 0, \sum_{k=1}^{d-1} a_k}_{i}, \cdots, 0),  i \in \{ 1,\cdots, d-1\};  .$ We assume that  all the upside-down flipped version of these rules still hold.%\protect\overrightarrow{\alpha^{\prime}}=(\alpha_{1-j}-\alpha_{-j}, \cdots, \alpha_{d-1-j}-\alpha_{-j}),  \alpha_k=\alpha_{k+td}, \alpha_0=0, 0\leq j \leq d-1.$ %\protect\overleftarrow{\alpha}=(\alpha_{d-1}, \cdots, \alpha_1).$
}\label{quditrules1}
  \end{figure}
 \FloatBarrier

  \begin{figure}[!h]
\begin{center}
\[
\quad \qquad\begin{array}{|cccc|}
\hline
	\beginpgfgraphicnamed{TikZit/triangleocopydit}
	\InputIfFileExists{TikZit/triangleocopydit.tikz}{}{\input{./figures/TikZit/triangleocopydit.tikz}}%
	\endpgfgraphicnamed
 &(Bs0) &%
	\beginpgfgraphicnamed{TikZit/trianglepicopydit}
	\InputIfFileExists{TikZit/trianglepicopydit.tikz}{}{\input{./figures/TikZit/trianglepicopydit.tikz}}%
	\endpgfgraphicnamed
&(Bsj)\\
	\beginpgfgraphicnamed{TikZit/sucdit}
	\InputIfFileExists{TikZit/sucdit.tikz}{}{\input{./figures/TikZit/sucdit.tikz}}%
	\endpgfgraphicnamed
&(Suc)& %
	\beginpgfgraphicnamed{TikZit/triangleinvers}
	\InputIfFileExists{TikZit/triangleinvers.tikz}{}{\input{./figures/TikZit/triangleinvers.tikz}}%
	\endpgfgraphicnamed
  & (Inv) \\
   & &&\\
	\beginpgfgraphicnamed{TikZit/phasecopydit}
	\InputIfFileExists{TikZit/phasecopydit.tikz}{}{\input{./figures/TikZit/phasecopydit.tikz}}%
	\endpgfgraphicnamed
&(Pcy)& %
	\beginpgfgraphicnamed{TikZit/additiondit}
	\InputIfFileExists{TikZit/additiondit.tikz}{}{\input{./figures/TikZit/additiondit.tikz}}%
	\endpgfgraphicnamed
&(AD) \\
	\beginpgfgraphicnamed{TikZit/wsymetrydit}
	\InputIfFileExists{TikZit/wsymetrydit.tikz}{}{\input{./figures/TikZit/wsymetrydit.tikz}}%
	\endpgfgraphicnamed
&(Sym) &  %
	\beginpgfgraphicnamed{TikZit/associatedit}
	\InputIfFileExists{TikZit/associatedit.tikz}{}{\input{./figures/TikZit/associatedit.tikz}}%
	\endpgfgraphicnamed
 &(Aso)\\ 
   & &&\\
	\beginpgfgraphicnamed{TikZit/whopfdit}
	\InputIfFileExists{TikZit/whopfdit.tikz}{}{\input{./figures/TikZit/whopfdit.tikz}}%
	\endpgfgraphicnamed
&(Whf) &  %
	\beginpgfgraphicnamed{TikZit/brkdit}
	\InputIfFileExists{TikZit/brkdit.tikz}{}{\input{./figures/TikZit/brkdit.tikz}}%
	\endpgfgraphicnamed
 &(Brk)\\ 
	\beginpgfgraphicnamed{TikZit/dboxtrianglekt}
	\InputIfFileExists{TikZit/dboxtrianglekt.tikz}{}{\input{./figures/TikZit/dboxtrianglekt.tikz}}%
	\endpgfgraphicnamed
&(DT) &  %
	\beginpgfgraphicnamed{TikZit/2trianglebw2gnlmdm}
	\InputIfFileExists{TikZit/2trianglebw2gnlmdm.tikz}{}{\input{./figures/TikZit/2trianglebw2gnlmdm.tikz}}%
	\endpgfgraphicnamed
 &(Tre)\\ 
	\beginpgfgraphicnamed{TikZit/trianglekj}
	\InputIfFileExists{TikZit/trianglekj.tikz}{}{\input{./figures/TikZit/trianglekj.tikz}}%
	\endpgfgraphicnamed
&(TKj) &    %
	\beginpgfgraphicnamed{TikZit/trianglecopyvariant1}
	\InputIfFileExists{TikZit/trianglecopyvariant1.tikz}{}{\input{./figures/TikZit/trianglecopyvariant1.tikz}}%
	\endpgfgraphicnamed
&(Brk2) \\ 
%\tikzfig{TikZit/whopfdit}&(Whf) &  &\\ 
  		  		\hline
  		\end{array}\]  
  	\end{center}
  	\caption{Qudit ZX-calculus rules II,  where $\protect\overrightarrow{1}=\protect\overbrace{(1,\cdots,1)}^{d-1};\protect\overrightarrow{a}=(a_1,\cdots, a_{d-1}); \protect\overrightarrow{b}=(b_1,\cdots, b_{d-1}); a_k, b_k\in  \mathbb C, k \in \{1,\cdots, d-1\}; V_j=\protect\overbrace{(\protect\underbrace{0,\cdots, 1}_{d-j}, \cdots, 0)}^{d-1}; j \in \{ 1,\cdots, d-1\}.$}\label{quditrules2}
  \end{figure}
 \FloatBarrier
 \begin{remark}
The rule (EU) given in Figure \ref{quditrules1}  was essentially found (without scalars) by KangFeng Ng and the author as a generalisation of the Euler decomposition of the Hadamard gate. This rule has been reported in several talks, e.g., in QPL 2019 \cite{qwangslides19}.  
 \end{remark}

   \section{Properties of qudit ZX-calculus}
   In this section, we show that given the above generators and rewriting rules of qudit ZX-calculus, what kind of properties it would have. Theses properties will be presented in terms of lemmas, corollaries and propositions.  
   
   First, we show that why the transpose of the triangle node is well-defined.
    \begin{lemma}\label{dtriangleequvditlm}
	\beginpgfgraphicnamed{TikZit/dtriangleequvdit}
	\InputIfFileExists{TikZit/dtriangleequvdit.tikz}{}{\input{./figures/TikZit/dtriangleequvdit.tikz}}%
	\endpgfgraphicnamed

   \end{lemma}
 \begin{proof}
\[    %
	\beginpgfgraphicnamed{TikZit/dtriangleequvditprf}
	\InputIfFileExists{TikZit/dtriangleequvditprf.tikz}{}{\input{./figures/TikZit/dtriangleequvditprf.tikz}}%
	\endpgfgraphicnamed
\]
\end{proof}

   \begin{lemma}\label{controltrianleditlm}
\[    %
	\beginpgfgraphicnamed{TikZit/controltrianledit}
	\InputIfFileExists{TikZit/controltrianledit.tikz}{}{\input{./figures/TikZit/controltrianledit.tikz}}%
	\endpgfgraphicnamed
\]
   \end{lemma}
 \begin{proof}
\[    %
	\beginpgfgraphicnamed{TikZit/controltrianleditprf}
	\InputIfFileExists{TikZit/controltrianleditprf.tikz}{}{\input{./figures/TikZit/controltrianleditprf.tikz}}%
	\endpgfgraphicnamed
\]
  The second equality can be proved similarly. 
\end{proof} 
 This lemma justifies the usage of a triangle node on a horizontal wire as used in rule (Brk2).

      \begin{lemma}\label{kjgalm}
  \[  %
	\beginpgfgraphicnamed{TikZit/kjga}
	\InputIfFileExists{TikZit/kjga.tikz}{}{\input{./figures/TikZit/kjga.tikz}}%
	\endpgfgraphicnamed
\]
    where $ \overrightarrow{a}=(a_1,\cdots, a_{d-1}),  j \in \{ 1,\cdots, d-1\}.$%\overrightarrow{b}=\overbrace{(\underbrace{0,\cdots, a_{d-j}-1}_{d-j}, \cdots, 0)}^{d-1}, .$
   \end{lemma}
     \begin{proof}
      \[  %
	\beginpgfgraphicnamed{TikZit/kjgaprf}
	\InputIfFileExists{TikZit/kjgaprf.tikz}{}{\input{./figures/TikZit/kjgaprf.tikz}}%
	\endpgfgraphicnamed
\]
      \end{proof}
      
     \begin{lemma}\label{zeroemptyditlm}
	\beginpgfgraphicnamed{TikZit/zeroemptydit}
	\InputIfFileExists{TikZit/zeroemptydit.tikz}{}{\input{./figures/TikZit/zeroemptydit.tikz}}%
	\endpgfgraphicnamed

   \end{lemma}  
         \begin{proof}
      \[  %
	\beginpgfgraphicnamed{TikZit/zeroemptyditprf}
	\InputIfFileExists{TikZit/zeroemptyditprf.tikz}{}{\input{./figures/TikZit/zeroemptyditprf.tikz}}%
	\endpgfgraphicnamed
\]
      \end{proof}

    \begin{lemma}\label{scalarmultiplydlm}
  \[  %
	\beginpgfgraphicnamed{TikZit/scalarmultiplyd2}
	\InputIfFileExists{TikZit/scalarmultiplyd2.tikz}{}{\input{./figures/TikZit/scalarmultiplyd2.tikz}}%
	\endpgfgraphicnamed
\]
    where $\overrightarrow{s}=(0,\cdots,0, \frac{1}{d}-1)$. %$ \overrightarrow{u}=(u_1,\cdots, u_{d-1}), \overrightarrow{v}=(v_1,\cdots, v_{d-1}),  \overrightarrow{w}=\overbrace{(\underbrace{0,\cdots, 0}_{d-2},  -1+\sum_{i,j=0}^{d-1}u_iv_j)}^{d-1}, u_0=v_0=1.$
   \end{lemma}  
   
      \begin{proof}
      \[  %
	\beginpgfgraphicnamed{TikZit/scalarmultiplyd2prf}
	\InputIfFileExists{TikZit/scalarmultiplyd2prf.tikz}{}{\input{./figures/TikZit/scalarmultiplyd2prf.tikz}}%
	\endpgfgraphicnamed
\]
      \end{proof}
      
         \begin{lemma}\label{scalargeneralmultlm}
 \[   %
	\beginpgfgraphicnamed{TikZit/scalargeneralmult}
	\InputIfFileExists{TikZit/scalargeneralmult.tikz}{}{\input{./figures/TikZit/scalargeneralmult.tikz}}%
	\endpgfgraphicnamed
\]
   \end{lemma}  
       \begin{proof}
      \[  %
	\beginpgfgraphicnamed{TikZit/scalargeneralmultprf}
	\InputIfFileExists{TikZit/scalargeneralmultprf.tikz}{}{\input{./figures/TikZit/scalargeneralmultprf.tikz}}%
	\endpgfgraphicnamed
\]
      \end{proof}

         \begin{corollary}\label{dboxequalhlm}
  \[  %
	\beginpgfgraphicnamed{TikZit/dboxequalh}
	\InputIfFileExists{TikZit/dboxequalh.tikz}{}{\input{./figures/TikZit/dboxequalh.tikz}}%
	\endpgfgraphicnamed
\]
    where $\overrightarrow{s}=(0,\cdots,0, \frac{1}{d}-1)$. 
   \end{corollary}
      This follows directly from lemma \ref{scalarmultiplydlm} and the definition of the D box.
      
    \begin{lemma}\label{hilm}
	\beginpgfgraphicnamed{TikZit/h1scalar}
	\InputIfFileExists{TikZit/h1scalar.tikz}{}{\input{./figures/TikZit/h1scalar.tikz}}%
	\endpgfgraphicnamed

   \end{lemma}
    \begin{proof}
      \[  %
	\beginpgfgraphicnamed{TikZit/h1scalarprf}
	\InputIfFileExists{TikZit/h1scalarprf.tikz}{}{\input{./figures/TikZit/h1scalarprf.tikz}}%
	\endpgfgraphicnamed
\]
      \end{proof}

         \begin{corollary}\label{redspiderforgrlm}
  \[  %
	\beginpgfgraphicnamed{TikZit/redspiderforgr}
	\InputIfFileExists{TikZit/redspiderforgr.tikz}{}{\input{./figures/TikZit/redspiderforgr.tikz}}%
	\endpgfgraphicnamed
\]
    where $\overrightarrow{t}=(0,\cdots,0, \frac{1}{d^{m+n-1}}-1),  j \in \{ 0,1,\cdots, d-1\}$. 
   \end{corollary}
      This follows directly from lemma \ref{scalarmultiplydlm}, (H1) and the definition of the red spider, we also denote the equality as (H).  
   
The X spider fusion rule can also be derived.
    
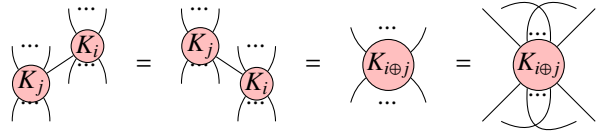
\begin{lemma}\label{s4lm}
\[    %
	\beginpgfgraphicnamed{TikZit/redspider0pfusedit2}
	\InputIfFileExists{TikZit/redspider0pfusedit2.tikz}{}{\input{./figures/TikZit/redspider0pfusedit2.tikz}}%
	\endpgfgraphicnamed
\]
    
   \end{lemma}   
  The proof follows directly from (S1), (H) and lemma \ref{scalarmultiplydlm} and\ref{hilm}. 
   
 Note that there is just one edge connected with two X spiders when they are fusing. How about if there are more than one edges appeared?

    
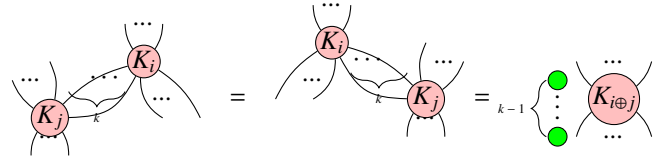
\begin{lemma}\label{redspidermwiresfuseditlm}
  \[  %
	\beginpgfgraphicnamed{TikZit/redspidermwiresfusedit}
	\InputIfFileExists{TikZit/redspidermwiresfusedit.tikz}{}{\input{./figures/TikZit/redspidermwiresfusedit.tikz}}%
	\endpgfgraphicnamed
\]
      \end{lemma} 
   The proof is similar to lemma \ref{s4lm}.

 \begin{lemma}\label{dboxsquarelm}
	\beginpgfgraphicnamed{TikZit/dboxsquare}
	\InputIfFileExists{TikZit/dboxsquare.tikz}{}{\input{./figures/TikZit/dboxsquare.tikz}}%
	\endpgfgraphicnamed

   \end{lemma}
     \begin{proof}
      \[  %
	\beginpgfgraphicnamed{TikZit/dboxsquareprf}
	\InputIfFileExists{TikZit/dboxsquareprf.tikz}{}{\input{./figures/TikZit/dboxsquareprf.tikz}}%
	\endpgfgraphicnamed
\]
      \end{proof}
   
    \begin{lemma}\label{dboxgcopylm}
 \[   %
	\beginpgfgraphicnamed{TikZit/dboxgcopy}
	\InputIfFileExists{TikZit/dboxgcopy.tikz}{}{\input{./figures/TikZit/dboxgcopy.tikz}}%
	\endpgfgraphicnamed
\]
   \end{lemma}
     \begin{proof}
     \[  %
	\beginpgfgraphicnamed{TikZit/dboxgcopyprf}
	\InputIfFileExists{TikZit/dboxgcopyprf.tikz}{}{\input{./figures/TikZit/dboxgcopyprf.tikz}}%
	\endpgfgraphicnamed
\]
        The other equalities can be proved similarly. 
      \end{proof}
      
     \begin{lemma}\label{dboxgdotlm}
 \[   %
	\beginpgfgraphicnamed{TikZit/dboxgdot}
	\InputIfFileExists{TikZit/dboxgdot.tikz}{}{\input{./figures/TikZit/dboxgdot.tikz}}%
	\endpgfgraphicnamed
\]
   \end{lemma}
    \begin{proof}
     \[  %
	\beginpgfgraphicnamed{TikZit/dboxgdotprf}
	\InputIfFileExists{TikZit/dboxgdotprf.tikz}{}{\input{./figures/TikZit/dboxgdotprf.tikz}}%
	\endpgfgraphicnamed
\]
      The second equality can be proved similarly. 
      \end{proof}

        \begin{lemma}\label{dboxonkjlm}
 \[   %
	\beginpgfgraphicnamed{TikZit/dboxonkjgt}
	\InputIfFileExists{TikZit/dboxonkjgt.tikz}{}{\input{./figures/TikZit/dboxonkjgt.tikz}}%
	\endpgfgraphicnamed
\]
 where  $j \in \{ 0,1,\cdots, d-1\}$.
   \end{lemma}

        \begin{lemma}\label{dboxonkjlm}
 \[   %
	\beginpgfgraphicnamed{TikZit/dboxonkj}
	\InputIfFileExists{TikZit/dboxonkj.tikz}{}{\input{./figures/TikZit/dboxonkj.tikz}}%
	\endpgfgraphicnamed
\]
 where  $j \in \{ 0,1,\cdots, d-1\}$.
   \end{lemma}   
     \begin{proof}
     \[  %
	\beginpgfgraphicnamed{TikZit/dboxonkjprf}
	\InputIfFileExists{TikZit/dboxonkjprf.tikz}{}{\input{./figures/TikZit/dboxonkjprf.tikz}}%
	\endpgfgraphicnamed
\]
      \end{proof}

      \begin{lemma}\label{dualiserslm}
	\beginpgfgraphicnamed{TikZit/dualisers}
	\InputIfFileExists{TikZit/dualisers.tikz}{}{\input{./figures/TikZit/dualisers.tikz}}%
	\endpgfgraphicnamed

   \end{lemma}
     \begin{proof}
     \[  %
	\beginpgfgraphicnamed{TikZit/dualisersprf1}
	\InputIfFileExists{TikZit/dualisersprf1.tikz}{}{\input{./figures/TikZit/dualisersprf1.tikz}}%
	\endpgfgraphicnamed
\]
      \end{proof}
      
   \begin{lemma}\label{hadslidegnlm}
\[    %
	\beginpgfgraphicnamed{TikZit/hadslidegn}
	\InputIfFileExists{TikZit/hadslidegn.tikz}{}{\input{./figures/TikZit/hadslidegn.tikz}}%
	\endpgfgraphicnamed
\]
   \end{lemma}    
       \begin{proof}
     \[  %
	\beginpgfgraphicnamed{TikZit/hadslidegnprf}
	\InputIfFileExists{TikZit/hadslidegnprf.tikz}{}{\input{./figures/TikZit/hadslidegnprf.tikz}}%
	\endpgfgraphicnamed
\]
     The other equalities can be proved similarly. 
      \end{proof}
  
     \begin{corollary}\label{dboxslidegrnlm}
	\beginpgfgraphicnamed{TikZit/dboxslidegrn}
	\InputIfFileExists{TikZit/dboxslidegrn.tikz}{}{\input{./figures/TikZit/dboxslidegrn.tikz}}%
	\endpgfgraphicnamed

   \end{corollary}    
   \begin{proof}
     \[  %
	\beginpgfgraphicnamed{TikZit/dboxslidegrnprf}
	\InputIfFileExists{TikZit/dboxslidegrnprf.tikz}{}{\input{./figures/TikZit/dboxslidegrnprf.tikz}}%
	\endpgfgraphicnamed
\]
      The second equality can be proved similarly. 
      \end{proof}

  \begin{corollary}\label{rkjslidegnlm}
\[    %
	\beginpgfgraphicnamed{TikZit/rkjslidegn}
	\InputIfFileExists{TikZit/rkjslidegn.tikz}{}{\input{./figures/TikZit/rkjslidegn.tikz}}%
	\endpgfgraphicnamed
\]
where $ j \in \{ 0,1,\cdots, d-1\}$. 
   \end{corollary}    
     \begin{proof}
     \[  %
	\beginpgfgraphicnamed{TikZit/rkjslidegnprf}
	\InputIfFileExists{TikZit/rkjslidegnprf.tikz}{}{\input{./figures/TikZit/rkjslidegnprf.tikz}}%
	\endpgfgraphicnamed
\]
      \end{proof}

      \begin{lemma}\label{swaprgcapcupditlm}
\[    %
	\beginpgfgraphicnamed{TikZit/swaprgcapcupdit}
	\InputIfFileExists{TikZit/swaprgcapcupdit.tikz}{}{\input{./figures/TikZit/swaprgcapcupdit.tikz}}%
	\endpgfgraphicnamed
\]
 \end{lemma}         
 
    \begin{proof}
    We only prove the first equality, the others can be proved similarly.
     \[  %
	\beginpgfgraphicnamed{TikZit/swaprgcapcupditprf}
	\InputIfFileExists{TikZit/swaprgcapcupditprf.tikz}{}{\input{./figures/TikZit/swaprgcapcupditprf.tikz}}%
	\endpgfgraphicnamed
\]
      \end{proof}

     \begin{lemma}\label{grconnectkslm}
\[    %
	\beginpgfgraphicnamed{TikZit/grconnectks}
	\InputIfFileExists{TikZit/grconnectks.tikz}{}{\input{./figures/TikZit/grconnectks.tikz}}%
	\endpgfgraphicnamed
\]

   \end{lemma}      
      \begin{proof}
   \[  %
	\beginpgfgraphicnamed{TikZit/grconnectksprf}
	\InputIfFileExists{TikZit/grconnectksprf.tikz}{}{\input{./figures/TikZit/grconnectksprf.tikz}}%
	\endpgfgraphicnamed
\]
where    $\overrightarrow{s}=(0,\cdots,0, \frac{1}{d}-1), \overrightarrow{t_1}=(0,\cdots,0, \frac{1}{d^{d-1}}-1)$.
      \end{proof} 

Next we show that it doesn't make sense anymore to draw a horizontal line between Z spider and X spider in higher dimensional ZX-calculus ($d\geq 3$), which would be one of the main difference in comparison to qubit ZX-calculus.        
\begin{lemma}\label{controlnotslideditlm}
\[    %
	\beginpgfgraphicnamed{TikZit/controlnotslidedit}
	\InputIfFileExists{TikZit/controlnotslidedit.tikz}{}{\input{./figures/TikZit/controlnotslidedit.tikz}}%
	\endpgfgraphicnamed
\]
   \end{lemma}      
      
  \begin{proof}
We only prove the first equality, the others can be proved similarly.
\[    %
	\beginpgfgraphicnamed{TikZit/controlnotslideditprf}
	\InputIfFileExists{TikZit/controlnotslideditprf.tikz}{}{\input{./figures/TikZit/controlnotslideditprf.tikz}}%
	\endpgfgraphicnamed
\]
\end{proof}
     \begin{remark}
     With lemma \ref{controlnotslideditlm}, now we can use cap ($C_a$) and cup ($C_u$) to transpose simultaneously the diagrams on each side of the rules listed in Figure  \ref{quditrules2}, so that the upside-down flipped version of those rules still hold, except the rule  (Bsj) which has the form %
	\beginpgfgraphicnamed{TikZit/trianglepicopyditflip}
	\InputIfFileExists{TikZit/trianglepicopyditflip.tikz}{}{\input{./figures/TikZit/trianglepicopyditflip.tikz}}%
	\endpgfgraphicnamed
~ after transpose. We will use the same name for the flipped version of those rules.  
 \end{remark}
 
\begin{lemma}\label{slidecupditlm}
\[    %
	\beginpgfgraphicnamed{TikZit/slidecupdit}
	\InputIfFileExists{TikZit/slidecupdit.tikz}{}{\input{./figures/TikZit/slidecupdit.tikz}}%
	\endpgfgraphicnamed
\]
 where $\overrightarrow{a}=(a_1,\cdots, a_{d-1}), \overleftarrow{a}=(a_{d-1}, \cdots, a_1)$.
   \end{lemma}      
  \begin{proof}
\[    %
	\beginpgfgraphicnamed{TikZit/slidecupditprf}
	\InputIfFileExists{TikZit/slidecupditprf.tikz}{}{\input{./figures/TikZit/slidecupditprf.tikz}}%
	\endpgfgraphicnamed
\]
The second equality can be proved similarly as for the first one.
\end{proof}

  
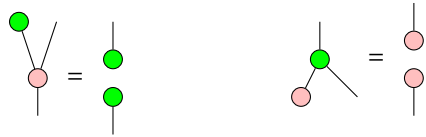
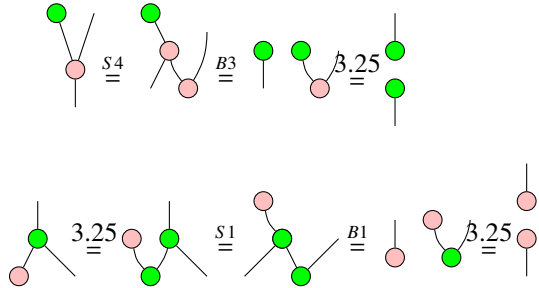
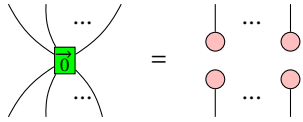
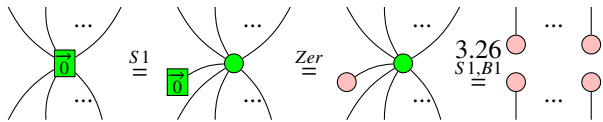
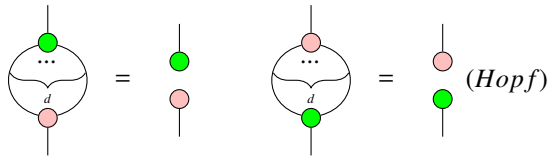
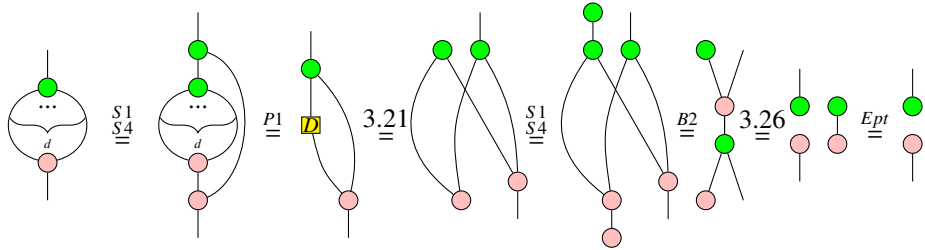
\begin{lemma}\label{copyvarsditlm}
\[    %
	\beginpgfgraphicnamed{TikZit/copyvarsdit}
	\InputIfFileExists{TikZit/copyvarsdit.tikz}{}{\input{./figures/TikZit/copyvarsdit.tikz}}%
	\endpgfgraphicnamed
\]
   \end{lemma}      
   \begin{proof}
\[    %
	\beginpgfgraphicnamed{TikZit/copyvarsditprf}
	\InputIfFileExists{TikZit/copyvarsditprf.tikz}{}{\input{./figures/TikZit/copyvarsditprf.tikz}}%
	\endpgfgraphicnamed
\]
\end{proof}  

 \begin{lemma}\label{spider0tordotslm}
 \[   %
	\beginpgfgraphicnamed{TikZit/spider0tordots}
	\InputIfFileExists{TikZit/spider0tordots.tikz}{}{\input{./figures/TikZit/spider0tordots.tikz}}%
	\endpgfgraphicnamed
\]
   \end{lemma}  
  \begin{proof}
\[ %
	\beginpgfgraphicnamed{TikZit/spider0tordotsprf}
	\InputIfFileExists{TikZit/spider0tordotsprf.tikz}{}{\input{./figures/TikZit/spider0tordotsprf.tikz}}%
	\endpgfgraphicnamed
  \]
  \end{proof}

  \begin{lemma}\label{hopfditlm}
\[    %
	\beginpgfgraphicnamed{TikZit/hopfdit}
	\InputIfFileExists{TikZit/hopfdit.tikz}{}{\input{./figures/TikZit/hopfdit.tikz}}%
	\endpgfgraphicnamed
  ~~(Hopf)\]
   \end{lemma}  
  \begin{proof}
\[    %
	\beginpgfgraphicnamed{TikZit/hopfditprf}
	\InputIfFileExists{TikZit/hopfditprf.tikz}{}{\input{./figures/TikZit/hopfditprf.tikz}}%
	\endpgfgraphicnamed
\]
The second equality can be proved similarly.
\end{proof}

\begin{corollary}\label{redgreenmchangelm}
\[    %
	\beginpgfgraphicnamed{TikZit/redgreenmchange}
	\InputIfFileExists{TikZit/redgreenmchange.tikz}{}{\input{./figures/TikZit/redgreenmchange.tikz}}%
	\endpgfgraphicnamed
\]
where $1\leq k \leq d-1$.
   \end{corollary}            
\begin{proof}
\[    %
	\beginpgfgraphicnamed{TikZit/redgreenmchangeprf}
	\InputIfFileExists{TikZit/redgreenmchangeprf.tikz}{}{\input{./figures/TikZit/redgreenmchangeprf.tikz}}%
	\endpgfgraphicnamed
\]
\end{proof}      
      
%  \begin{lemma}\label{kjcommutetrianglelm}
%    \tikzfig{TikZit/kjcommutetriangle}
%    where $1\leq j \leq d-1$.
%   \end{lemma}      
%  \begin{proof}
%\[    \tikzfig{TikZit/kjcommutetriangleprf}\]
%\end{proof}          

 \begin{lemma}\label{triangleonreddotlm}
	\beginpgfgraphicnamed{TikZit/triangleonreddot}
	\InputIfFileExists{TikZit/triangleonreddot.tikz}{}{\input{./figures/TikZit/triangleonreddot.tikz}}%
	\endpgfgraphicnamed
 
 \end{lemma}  
    \begin{proof}
	\beginpgfgraphicnamed{TikZit/triangleonreddotprfdit}
	\InputIfFileExists{TikZit/triangleonreddotprfdit.tikz}{}{\input{./figures/TikZit/triangleonreddotprfdit.tikz}}%
	\endpgfgraphicnamed
  
   \end{proof}  

      \begin{lemma}\label{sucinvditlm}
 \[   %
	\beginpgfgraphicnamed{TikZit/sucinvdit}
	\InputIfFileExists{TikZit/sucinvdit.tikz}{}{\input{./figures/TikZit/sucinvdit.tikz}}%
	\endpgfgraphicnamed
\]
   \end{lemma}
    \begin{proof}
	\beginpgfgraphicnamed{TikZit/sucinvditprf}
	\InputIfFileExists{TikZit/sucinvditprf.tikz}{}{\input{./figures/TikZit/sucinvditprf.tikz}}%
	\endpgfgraphicnamed
  
   \end{proof}  

  \begin{lemma}\label{Hopftdit}
  \begin{equation*}\label{TR4geq}
	\beginpgfgraphicnamed{TikZit/tr4g2}
	\InputIfFileExists{TikZit/tr4g2.tikz}{}{\input{./figures/TikZit/tr4g2.tikz}}%
	\endpgfgraphicnamed
 
   \end{equation*}
    \end{lemma}
 \begin{proof}
 \[%
	\beginpgfgraphicnamed{TikZit/tr4gprf2dit}
	\InputIfFileExists{TikZit/tr4gprf2dit.tikz}{}{\input{./figures/TikZit/tr4gprf2dit.tikz}}%
	\endpgfgraphicnamed
\] 
  The second equality can be proved similarly. 
  \end{proof}  
  
     \begin{lemma}\label{Hopfgtr}
  \begin{equation*}
	\beginpgfgraphicnamed{TikZit/trianglehopfgreen2}
	\InputIfFileExists{TikZit/trianglehopfgreen2.tikz}{}{\input{./figures/TikZit/trianglehopfgreen2.tikz}}%
	\endpgfgraphicnamed

  \end{equation*}
    \end{lemma}
 \begin{proof}
\[ %
	\beginpgfgraphicnamed{TikZit/trianglehopfgreenprfdit}
	\InputIfFileExists{TikZit/trianglehopfgreenprfdit.tikz}{}{\input{./figures/TikZit/trianglehopfgreenprfdit.tikz}}%
	\endpgfgraphicnamed
  \]
 The second equality can be proved similarly. 
  \end{proof}

       \begin{lemma}\label{trianglehopfditlm}
	\beginpgfgraphicnamed{TikZit/trianglehopfdit}
	\InputIfFileExists{TikZit/trianglehopfdit.tikz}{}{\input{./figures/TikZit/trianglehopfdit.tikz}}%
	\endpgfgraphicnamed

   \end{lemma}
 \begin{proof}
\[    %
	\beginpgfgraphicnamed{TikZit/trianglehopfditprf}
	\InputIfFileExists{TikZit/trianglehopfditprf.tikz}{}{\input{./figures/TikZit/trianglehopfditprf.tikz}}%
	\endpgfgraphicnamed
\]
  The second equality can be proved similarly. 
\end{proof}    
     
 \begin{corollary}\label{trianglehopfdit2lm}
\[    %
	\beginpgfgraphicnamed{TikZit/trianglehopfdit2}
	\InputIfFileExists{TikZit/trianglehopfdit2.tikz}{}{\input{./figures/TikZit/trianglehopfdit2.tikz}}%
	\endpgfgraphicnamed
\]
   \end{corollary}           
  
    \begin{lemma}\label{cnotlikemovelm}
\[    %
	\beginpgfgraphicnamed{TikZit/cnotlikemove}
	\InputIfFileExists{TikZit/cnotlikemove.tikz}{}{\input{./figures/TikZit/cnotlikemove.tikz}}%
	\endpgfgraphicnamed
\]
   \end{lemma}
    \begin{proof}
\[    %
	\beginpgfgraphicnamed{TikZit/cnotlikemoveprf}
	\InputIfFileExists{TikZit/cnotlikemoveprf.tikz}{}{\input{./figures/TikZit/cnotlikemoveprf.tikz}}%
	\endpgfgraphicnamed
\]
  The second and the third equalities can be proved similarly. 
\end{proof}

 
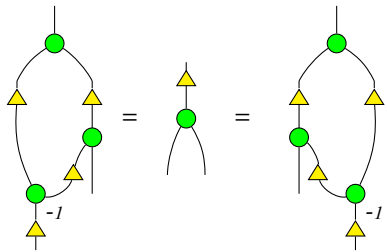
\begin{lemma}
\[ %
	\beginpgfgraphicnamed{TikZit/trianglecopylr}
	\InputIfFileExists{TikZit/trianglecopylr.tikz}{}{\input{./figures/TikZit/trianglecopylr.tikz}}%
	\endpgfgraphicnamed
 \]
  \end{lemma}   
 \begin{proof}
 The second equality follows directly from rules (Brk2) and (Inv). We only prove the first equality.
 \[%
	\beginpgfgraphicnamed{TikZit/trianglecopylrprfdit}
	\InputIfFileExists{TikZit/trianglecopylrprfdit.tikz}{}{\input{./figures/TikZit/trianglecopylrprfdit.tikz}}%
	\endpgfgraphicnamed
  \]
  \end{proof}     
 
If we do partial transpose (bending wires) on both sides of the rule (Brk2), then we get
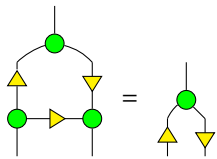
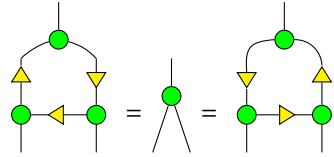
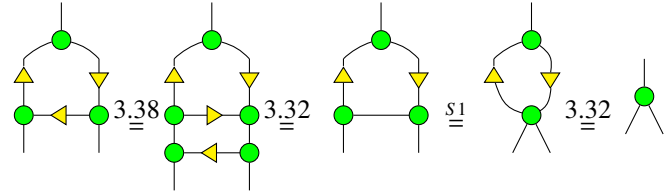
\begin{lemma}\label{brk2transpose}
\[ %
	\beginpgfgraphicnamed{TikZit/trianglecopyvariant2}
	\InputIfFileExists{TikZit/trianglecopyvariant2.tikz}{}{\input{./figures/TikZit/trianglecopyvariant2.tikz}}%
	\endpgfgraphicnamed
\]
 \end{lemma}

   \begin{lemma}\label{trianglecirc2ditlm}
 \[%
	\beginpgfgraphicnamed{TikZit/trianglecirc2}
	\InputIfFileExists{TikZit/trianglecirc2.tikz}{}{\input{./figures/TikZit/trianglecirc2.tikz}}%
	\endpgfgraphicnamed
\]
 \end{lemma}
 
  \begin{proof}
 % By (TR19), we have 
%$$  \tikzfig{TikZit/trianglecopyvariant1} $$
%After partial transpose, we have  
%$$  \tikzfig{TikZit/trianglecopyvariant2} $$
%Therefore, 
$$  %
	\beginpgfgraphicnamed{TikZit/trianglecopyvariant3dit}
	\InputIfFileExists{TikZit/trianglecopyvariant3dit.tikz}{}{\input{./figures/TikZit/trianglecopyvariant3dit.tikz}}%
	\endpgfgraphicnamed
 $$
The other part of the equality can be obtained by symmetry. 
 \end{proof}
 
    
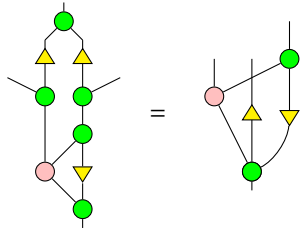
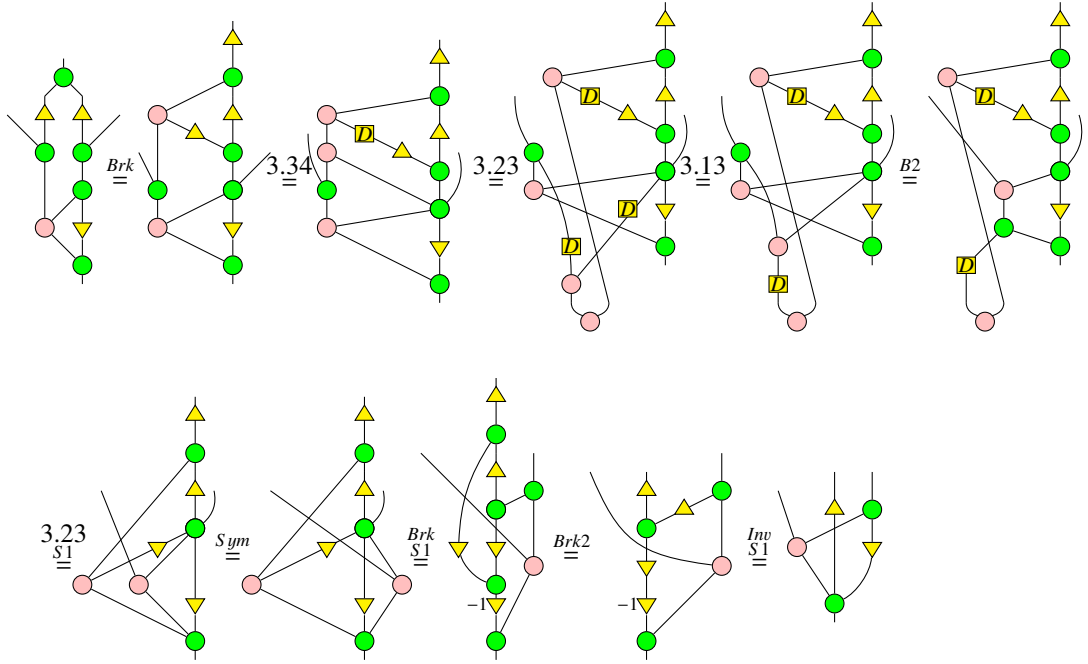
\begin{lemma}\label{genrealaiitionditlm}
\[    %
	\beginpgfgraphicnamed{TikZit/genrealaiitiondit}
	\InputIfFileExists{TikZit/genrealaiitiondit.tikz}{}{\input{./figures/TikZit/genrealaiitiondit.tikz}}%
	\endpgfgraphicnamed
\]
   \end{lemma}
    \begin{proof}
 \[%
	\beginpgfgraphicnamed{TikZit/genrealaiitionditprf}
	\InputIfFileExists{TikZit/genrealaiitionditprf.tikz}{}{\input{./figures/TikZit/genrealaiitionditprf.tikz}}%
	\endpgfgraphicnamed
  \]
  \end{proof}    
 
   \begin{lemma} \label{equivalentaddrulensditlm}
   Let $\overrightarrow{a}=(a_1,\cdots, a_{d-1}), \overrightarrow{b}=(b_1,\cdots, b_{d-1}), a_k, b_k\in  \mathbb C, k \in \{1,\cdots, d-1\}. $
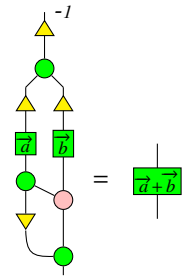
   \begin{equation*}
	\beginpgfgraphicnamed{TikZit/equivalentaddrulensdit}
	\InputIfFileExists{TikZit/equivalentaddrulensdit.tikz}{}{\input{./figures/TikZit/equivalentaddrulensdit.tikz}}%
	\endpgfgraphicnamed
 
   \end{equation*}    
     \end{lemma}
      \begin{proof}
   This lemma follows directly from lemma \ref{genrealaiitionditlm} and the (AD) rule, here we just give another proof.
 \[
	\beginpgfgraphicnamed{TikZit/equivalentaddrulensditprf2}
	\InputIfFileExists{TikZit/equivalentaddrulensditprf2.tikz}{}{\input{./figures/TikZit/equivalentaddrulensditprf2.tikz}}%
	\endpgfgraphicnamed
 
  \]
 % For the cases where both  $  \overrightarrow{a} $ $  \overrightarrow{b} $ are not invertible and are not of the same type (positions of non-zero elements), we don't know the proof yet.
  \end{proof}

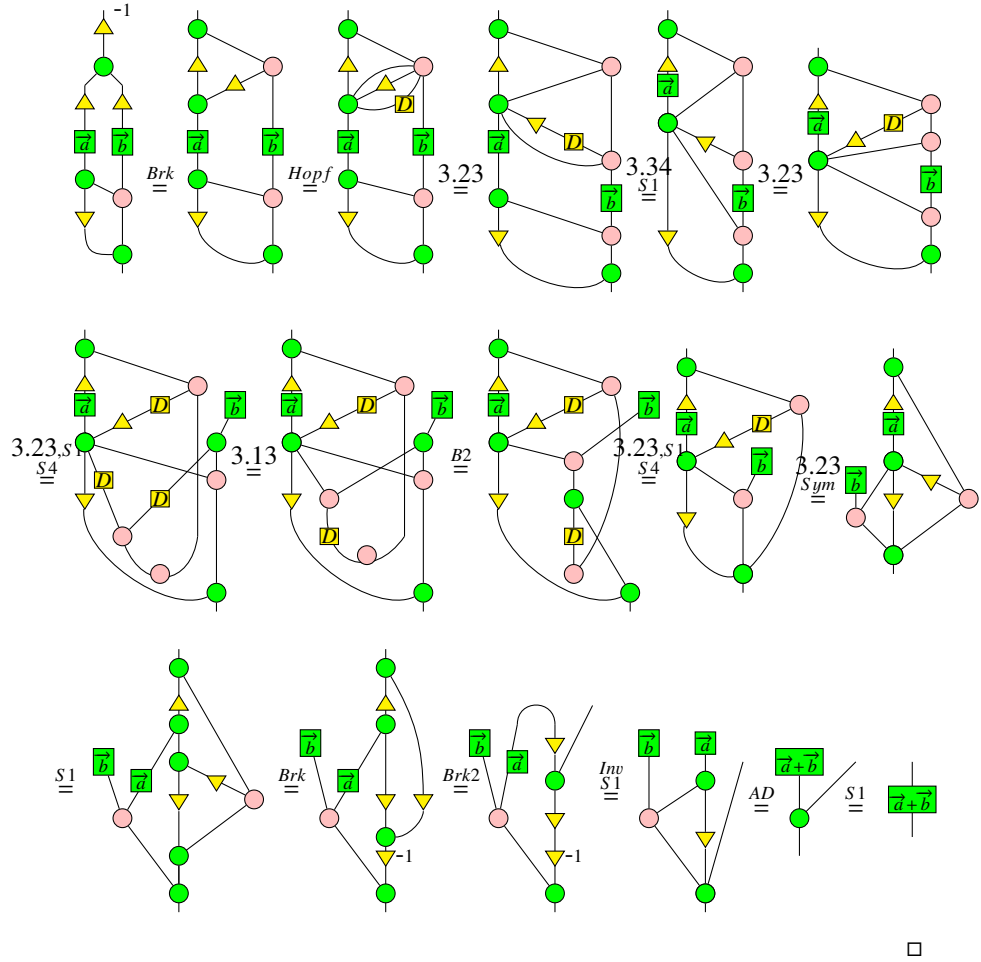
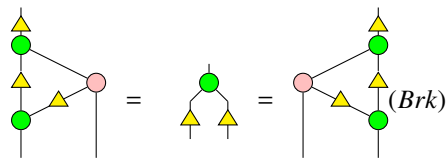
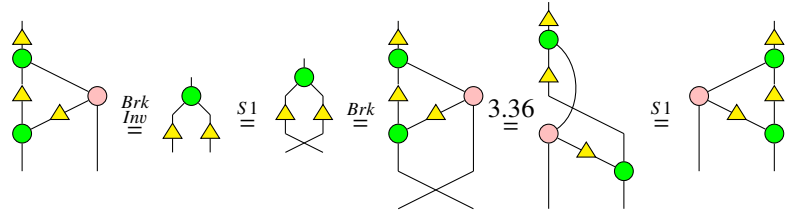
   \begin{lemma}
\[ %
	\beginpgfgraphicnamed{TikZit/brkvariantdit}
	\InputIfFileExists{TikZit/brkvariantdit.tikz}{}{\input{./figures/TikZit/brkvariantdit.tikz}}%
	\endpgfgraphicnamed
   (Brk) \]
    \end{lemma}
     \begin{proof}
\[  %
	\beginpgfgraphicnamed{TikZit/brkvariantditprf}
	\InputIfFileExists{TikZit/brkvariantditprf.tikz}{}{\input{./figures/TikZit/brkvariantditprf.tikz}}%
	\endpgfgraphicnamed
  \]
  We  call this derived equality (Brk) as well, since it is a variant of the  (Brk).
   \end{proof} 
   
  %   \begin{lemma}\label{symmetrykwireslm}
   %  For $1\leq k \leq d-1$, we have
%\[ \tikzfig{TikZit/symmetrykwires} \]
 %  \end{lemma}  

   \begin{lemma}\label{2kprf}
\[ %
	\beginpgfgraphicnamed{TikZit/2triangleup}
	\InputIfFileExists{TikZit/2triangleup.tikz}{}{\input{./figures/TikZit/2triangleup.tikz}}%
	\endpgfgraphicnamed
 \]
   \end{lemma}  
 \begin{proof}
$$  %
	\beginpgfgraphicnamed{TikZit/2triangleupprfdit}
	\InputIfFileExists{TikZit/2triangleupprfdit.tikz}{}{\input{./figures/TikZit/2triangleupprfdit.tikz}}%
	\endpgfgraphicnamed
$$
  \end{proof}     
  
      \begin{lemma}
\[  %
	\beginpgfgraphicnamed{TikZit/anddflipwitha2dit}
	\InputIfFileExists{TikZit/anddflipwitha2dit.tikz}{}{\input{./figures/TikZit/anddflipwitha2dit.tikz}}%
	\endpgfgraphicnamed
 \quad(Brkp)      \]
  where $\overrightarrow{a}=(a_1,\cdots, a_{d-1})$.
   \end{lemma} 
   \begin{proof}
   \[  %
	\beginpgfgraphicnamed{TikZit/anddflipwitha2prfdit}
	\InputIfFileExists{TikZit/anddflipwitha2prfdit.tikz}{}{\input{./figures/TikZit/anddflipwitha2prfdit.tikz}}%
	\endpgfgraphicnamed
 \]
   Therefore,
      \[  %
	\beginpgfgraphicnamed{TikZit/anddflipwitha2prf2dit}
	\InputIfFileExists{TikZit/anddflipwitha2prf2dit.tikz}{}{\input{./figures/TikZit/anddflipwitha2prf2dit.tikz}}%
	\endpgfgraphicnamed
 \]
     \end{proof}   
   
  %  \begin{corollary} \label{andcopy}
    %\tikzfig{TikZit/andcopymet}
  %\end{corollary}  

  \begin{lemma}\label{1triangle1kjbw2gnditlm}
 Let  $1\leq j \leq d-1$. Then
  \[  %
	\beginpgfgraphicnamed{TikZit/1triangle1kjbw2gndit}
	\InputIfFileExists{TikZit/1triangle1kjbw2gndit.tikz}{}{\input{./figures/TikZit/1triangle1kjbw2gndit.tikz}}%
	\endpgfgraphicnamed
 \]
   \end{lemma}  
  \begin{proof}
\[ %
	\beginpgfgraphicnamed{TikZit/1triangle1kjbw2gnditprf}
	\InputIfFileExists{TikZit/1triangle1kjbw2gnditprf.tikz}{}{\input{./figures/TikZit/1triangle1kjbw2gnditprf.tikz}}%
	\endpgfgraphicnamed
 \]
  \end{proof}

 \begin{lemma}\label{1tricpto2redlmdit} 
\[ %
	\beginpgfgraphicnamed{TikZit/1tricpto2reddit}
	\InputIfFileExists{TikZit/1tricpto2reddit.tikz}{}{\input{./figures/TikZit/1tricpto2reddit.tikz}}%
	\endpgfgraphicnamed
  \]
 \end{lemma}
 \begin{proof}
\[ %
	\beginpgfgraphicnamed{TikZit/1tricpto2redditprf}
	\InputIfFileExists{TikZit/1tricpto2redditprf.tikz}{}{\input{./figures/TikZit/1tricpto2redditprf.tikz}}%
	\endpgfgraphicnamed
 \]
  \end{proof}

%   \begin{lemma}\label{k2nditlm}
%    \tikzfig{TikZit/k2ndit}
  % \end{lemma}

 \begin{lemma}\label{trianglepiinverselm}
 \[   %
	\beginpgfgraphicnamed{TikZit/trianglepiinverse}
	\InputIfFileExists{TikZit/trianglepiinverse.tikz}{}{\input{./figures/TikZit/trianglepiinverse.tikz}}%
	\endpgfgraphicnamed
\]
 where $\overrightarrow{-1}=\protect\overbrace{(-1,\cdots,-1)}^{d-1}$.
   \end{lemma}
    \begin{proof}
\[ %
	\beginpgfgraphicnamed{TikZit/trianglepiinverseprf}
	\InputIfFileExists{TikZit/trianglepiinverseprf.tikz}{}{\input{./figures/TikZit/trianglepiinverseprf.tikz}}%
	\endpgfgraphicnamed
 \]
  \end{proof}

   \section{Normal form for qudits}
    
  Suppose $\{e_{k}| 0\leq k \leq d^{m}-1\}$ are the $d^m$-dimensional standard unit  column vectors (with entries all 0s except for a single 1):
 \[
e_k=\begin{blockarray}{cl}
%a & b & c & d & e \\
\begin{block}{(c)l}
     0 &r_0\\
      \vdots    &  \\
      1&r_k\\
     \vdots &  \\
       0&r_{d^m-1}\\
\end{block}
\end{blockarray}
 \]
 where $r_i$ denote the $i$-th row, $0\leq i \leq d^{m}-1,  m\geq 1$.
 
 Let 
  \[
\ket{i}=\begin{blockarray}{cl}
\begin{block}{(c)l}
     0 &r_0\\
      \vdots    &  \\
      1&r_i\\
     \vdots &  \\
       0&r_{d-1}\\
\end{block}
\end{blockarray}
 \]
 where $r_i$ denotes the $i$-th row, $0\leq i \leq d-1$.

Then 
\begin{lemma}\label{qbitstovectorditlm}
\begin{equation}\label{qbitstovectordit}
$$\ket{a_{m-1}\cdots a_k \cdots a_0}=e_{\sum_{i=0}^{m-1}a_id^i},$$
\end{equation}
where $a_k\in \{0, 1, \cdots, d-1\}, 0\leq k \leq m-1, m\geq 1$.
\end{lemma}
\begin{proof}
We prove by induction on $m$. If $m=1$, then by the definition of $e_k$ and $\ket{i}$, we have $\ket{a_0}=e_{a_0}=e_{\sum_{i=0}^{m-1}a_id^i}.$ Suppose (\ref{qbitstovectordit}) holds for  $m \geq 1$, then 
\[
\begin{array}{c}
\ket{a_ma_{m-1}\cdots a_k \cdots a_0}=\ket{a_m}\otimes \ket{a_{m-1}\cdots a_k \cdots a_0}=\ket{a_m}\otimes e_{\sum_{i=0}^{m-1}a_id^i}\\
=\begin{blockarray}{cl}
\begin{block}{(c)l}
     0 &r_0\\
      \vdots    &  \\
      1&r_{a_m}\\
     \vdots &  \\
       0&r_{d-1}\\
\end{block}
\end{blockarray}\otimes e_{\sum_{i=0}^{m-1}a_id^i}=
\begin{blockarray}{cl}
\begin{block}{(c)l}
    \Large O &R_0\\
      \vdots    &  \\
     e_{\sum_{i=0}^{m-1}a_id^i}&R_{a_m}\\
     \vdots &  \\
       \Large O  &R_{d-1}\\
\end{block}
\end{blockarray}=e_{a_md^m+\sum_{i=0}^{m-1}a_id^i}=e_{\sum_{i=0}^{m}a_id^i},

\end{array}
\]
where $R_i$ denotes a column vector with $d^m$ elements.
 \end{proof}
 
%Before we construct a normal form, we first give a graphical representation of elementary row addition. 
\begin{lemma} Suppose  $0\leq j_1< \cdots < j_s \leq m-1, 1\leq s \leq m$. 
Then
 \begin{equation}\label{rowaddrepresentation}
 \left\llbracket%
	\beginpgfgraphicnamed{TikZit/rowaddrepresentdit}
	\InputIfFileExists{TikZit/rowaddrepresentdit.tikz}{}{\input{./figures/TikZit/rowaddrepresentdit.tikz}}%
	\endpgfgraphicnamed
\right\rrbracket
=\begin{blockarray}{cccccl}
\begin{block}{(ccccc)l}
     1 & \cdots & 0 &\cdots & 0 &r_0\\
     \vdots    & \ddots & &&  \vdots&  \\
        0   & \cdots & 1 & \cdots& a&r_l\\
       \vdots    & &  & \ddots&  \vdots &  \\
        0   & \cdots &0 & \cdots& 1&r_{d^m-1}\\
\end{block}
\end{blockarray}
 \end{equation}
where $\overrightarrow{a}=(0,\cdots, 0, a), a\in  \mathbb C$, the $\overrightarrow{a}$ node connects to  $ j_i $ with  $k_i$ wires via red dots below it, $1\leq k_i\leq d-1,$   $r_l$ denotes the $l$-th row, $l=d^m-1-(k_1d^{j_1}+\cdots+k_sd^{j_s})$.
 \end{lemma}
 \begin{proof}
 Denote by $A$ the $d^m\times d^m$ row-addition elementary matrix  
 in (\ref{rowaddrepresentation}). Let $\{A_{k}| 0\leq k \leq d^{m}-1\}$ be the set of the columns  of $A$. Then
 $$A_{k}=Ae_{k}=\left\{
\begin{array}{ll}
e_k, 0\leq k \leq d^{m}-2\\
e_{d^m-1}+ae_l, k = d^{m}-1\\
\end{array}
\right.$$ 
where 
 \[
A_{d^{m}-1}=\begin{blockarray}{cl}
\begin{block}{(c)l}
     0 &r_0\\
      \vdots    &  \\
      a&r_l\\
     \vdots &  \\
       1&r_{d^m-1}\\
\end{block}
\end{blockarray}
 \]
By the equality (\ref{qbitstovectordit}),  assume $e_k$ has the form $e_k=\ket{a_{m-1}\cdots a_i \cdots a_0}$, $a_i\in \{0, 1, \cdots, d-1\}, 0\leq i \leq m-1$. Clearly, if $0\leq k \leq d^{m}-2$, then there must exist some $a_t \in \{0, 1, \cdots, d-2\}$. Therefore, % If $a_t=0$, then $e_k$ can be represented by the  following diagram:
% $$ \tikzfig{TikZit/rowaddverify0dit}$$
\[
\left\llbracket%
	\beginpgfgraphicnamed{TikZit/rowaddverify0dit}
	\InputIfFileExists{TikZit/rowaddverify0dit.tikz}{}{\input{./figures/TikZit/rowaddverify0dit.tikz}}%
	\endpgfgraphicnamed
\right\rrbracket =e_k
\]
 where $i_t \in \{0, 1, \cdots, d-1 \}$ and $i_t\neq 1.$ Then  for $0\leq k \leq d^{m}-2$, we have
\[
Ae_k=\left\llbracket%
	\beginpgfgraphicnamed{TikZit/rowaddverify1dit}
	\InputIfFileExists{TikZit/rowaddverify1dit.tikz}{}{\input{./figures/TikZit/rowaddverify1dit.tikz}}%
	\endpgfgraphicnamed
\right\rrbracket =e_k
\]

% $$  \left\llbracket\tikzfig{TikZit/rowaddverify2}\right\rrbracket =\ket{\overbrace{1\cdots 1}^m}+a_j\ket{\underset{m-1}{1}\cdots \underset{j_s}{0}\cdots 1\cdots \underset{j_1}{0}\cdots \underset{0}{1}}$$
For $k = d^{m}-1$, we have
 $$\left\llbracket%
	\beginpgfgraphicnamed{TikZit/standardunitdit}
	\InputIfFileExists{TikZit/standardunitdit.tikz}{}{\input{./figures/TikZit/standardunitdit.tikz}}%
	\endpgfgraphicnamed
\right\rrbracket=e_{d^m-1}$$  

 Then
\[
Ae_k =\begin{array}{ll}
\left\llbracket%
	\beginpgfgraphicnamed{TikZit/rowaddverify2dit}
	\InputIfFileExists{TikZit/rowaddverify2dit.tikz}{}{\input{./figures/TikZit/rowaddverify2dit.tikz}}%
	\endpgfgraphicnamed
\right\rrbracket &\\ =
[\ket{\underset{m-1}{d-1}}\otimes\cdots \otimes\underset{j_s}{I_d}\cdots\otimes \ket{d-1}\cdots \otimes\underset{j_1}{I_d}\cdots \otimes\ket{\underset{0}{d-1}}]\circ
[\ket{d-1}^{\otimes s}+a\ket{(d-1-k_s),\cdots, (d-1-k_1)}]&\\

=\ket{d-1}^{\otimes m}+a\ket{\underset{m-1}{(d-1)}\cdots \underset{j_s}{(d-1-k_s)}\cdots (d-1)\cdots \underset{j_1}{(d-1-k_1)}\cdots \underset{0}{(d-1)}}&\\
=e_{d^m-1}+ae_{l} &
\end{array}
\]
where $l=d^m-1-(k_1d^{j_1}+\cdots+k_sd^{j_s})$, $I_d$ is the $d$ dimensional identity operator, and we used 
( \ref{qbitstovectordit}) for the last equality. %Therefore, we have shown that the diagram (\ref{rowaddrepresentation}) represents the matrix $A$.

\end{proof}
 Similarly, we can prove the following lemma. 
 \begin{lemma} Suppose  $0\leq j_1< \cdots < j_s \leq m-1, 1\leq s \leq m$. 
Then
 \begin{equation}\label{rowmultrepresentditlm}
 \left\llbracket%
	\beginpgfgraphicnamed{TikZit/rowmultrepresentdit}
	\InputIfFileExists{TikZit/rowmultrepresentdit.tikz}{}{\input{./figures/TikZit/rowmultrepresentdit.tikz}}%
	\endpgfgraphicnamed
\right\rrbracket
=\begin{blockarray}{cccccl}
\begin{block}{(ccccc)l}
     1 & \cdots & 0 &\cdots & 0 &r_0\\
     \vdots    & \ddots & &&  \vdots&  \\
        0   & \cdots & 1 & \cdots& 0&r_i\\
       \vdots    & &  & \ddots&  \vdots &  \\
        0   & \cdots &0 & \cdots& a&r_{d^m-1}\\
\end{block}
\end{blockarray}
 \end{equation}
where $\overrightarrow{a}=(1,\cdots, 1, a), a\in  \mathbb C$.
 \end{lemma}

Given an arbitrary vector as
$ \begin{pmatrix}
        a_0  \\  a_1\\
        \vdots \\ a_{d^m-2}\\
        a_{d^m-1} \end{pmatrix},$ we claim that it can be uniquely represented by the following normal form:  
         \begin{equation}\label{normalformdit}
	\beginpgfgraphicnamed{TikZit/quditnormalform}
	\InputIfFileExists{TikZit/quditnormalform.tikz}{}{\input{./figures/TikZit/quditnormalform.tikz}}%
	\endpgfgraphicnamed

 \end{equation}	
where $ \overrightarrow{a}_i=(0,\cdots, 0, a_i), 0\leq i \leq d^{m}-2, \overrightarrow{a}_{d^m-1}=(1,\cdots, 1, a_{d^m-1})$. $a_i, a_{d^m-1} \in  \mathbb C$. There are $(d-1)\binom{m}{1}+(d-1)^2\binom{m}{2} +\cdots +(d-1)^m \binom{m}{m}=d^m-1$ row additions in the normal form.

Like the qubit case, the normal form (\ref{normalformdit}) is obtained via the following processes:
\[
 \begin{pmatrix}
        0  \\ 0\\
        \vdots \\
        1 \end{pmatrix} \xrightarrow[\text{addition}]{\text{row}}  \begin{pmatrix}
        a_0  \\ 0\\
        \vdots \\
        1 \end{pmatrix} \xrightarrow[\text{addition}]{\text{row}}  \begin{pmatrix}
        a_0  \\  a_1\\
        \vdots \\ a_{d^m-2}\\
        1 \end{pmatrix}  \xrightarrow[\text{multiplication}]{\text{row}}   \begin{pmatrix}
        a_0  \\  a_1\\
        \vdots \\ a_{d^m-2}\\
        a_{d^m-1} \end{pmatrix}
\]
In the case of $m=0$, for any complex number $a$, its normal form is defined as
  $$ %
	\beginpgfgraphicnamed{TikZit//scalarnormdit}
	\begin{tikzpicture}
	\begin{pgfonlayer}{nodelayer}
		\node [style=gbox] (0) at (0, -0.25) {${\scriptstyle \overrightarrow{a}}$};
		\node [style=rn] (1) at (0, 0.5) {$K_1$};
	\end{pgfonlayer}
	\begin{pgfonlayer}{edgelayer}
		\draw (1) to (0);
	\end{pgfonlayer}
\end{tikzpicture}}%
	\endpgfgraphicnamed
$$
where 
 $$\overrightarrow{a}=(0,\cdots, 0, a), \left\llbracket%
	\beginpgfgraphicnamed{TikZit//scalarnormdit}
	\begin{tikzpicture}
	\begin{pgfonlayer}{nodelayer}
		\node [style=gbox] (0) at (0, -0.25) {${\scriptstyle \overrightarrow{a}}$};
		\node [style=rn] (1) at (0, 0.5) {$K_1$};
	\end{pgfonlayer}
	\begin{pgfonlayer}{edgelayer}
		\draw (1) to (0);
	\end{pgfonlayer}
\end{tikzpicture}}%
	\endpgfgraphicnamed
\right\rrbracket=a.$$  
By the map-state duality, we get the universality of qudit ZX-calculus over  $\mathbb{C}$:  any $d^m \times d^n$ matrix $A$ with $m, n \geq 0$ can be represented by a ZX diagram.

Furthermore, if we introduce the W spider given in \cite{qwangqditzw} as follows:
 $$ %
	\beginpgfgraphicnamed{TikZit//wntoalgzx}
	\InputIfFileExists{TikZit//wntoalgzx.tikz}{}{\input{./figures/TikZit//wntoalgzx.tikz}}%
	\endpgfgraphicnamed
 \quad\quad\quad\quad  %
	\beginpgfgraphicnamed{TikZit//mlegsblackspider}
	\InputIfFileExists{TikZit//mlegsblackspider.tikz}{}{\input{./figures/TikZit//mlegsblackspider.tikz}}%
	\endpgfgraphicnamed
$$
where \[ \left\llbracket %
	\beginpgfgraphicnamed{TikZit/mlegsblackspiderone}
	\InputIfFileExists{TikZit/mlegsblackspiderone.tikz}{}{\input{./figures/TikZit/mlegsblackspiderone.tikz}}%
	\endpgfgraphicnamed
 \right\rrbracket=\underbrace{\ket{0\cdots0}}_{m}\bra{0}+\sum_{i=1}^{d-1}\sum_{k=1}^{m}\overbrace{\ket{\underbrace{0\cdots 0}_{k-1} i 0\cdots 0}}^{m}\bra{i}, \]

then for an arbitrary vector 
$$ \begin{pmatrix}
        a_0  \\  a_1\\
        \vdots \\ a_{d^m-2}\\
        a_{d^m-1} \end{pmatrix},$$
  we have 
\begin{equation}\label{wstyleqduditnformeq}
   \begin{pmatrix}
        a_0  \\  a_1\\
        \vdots \\ a_{d^m-2}\\
        a_{d^m-1} \end{pmatrix}=\quad\quad\quad  \left\llbracket%
	\beginpgfgraphicnamed{TikZit//syw2}
	\InputIfFileExists{TikZit//syw2.tikz}{}{\input{./figures/TikZit//syw2.tikz}}%
	\endpgfgraphicnamed
\right\rrbracket
 \end{equation}
        
  where    $ \overrightarrow{a}_j=(0,\cdots, 0, a_j),  a_j  \in  \mathbb C, 0\leq j \leq d^{m}-1, j= \sum_{i=0}^{m-1}k_id^i, 0\leq k_i \leq d-1$, and the j-th green box $\overrightarrow{a}_j$ has connections to the i-th phase free pink node ($0\leq i \leq m-1$)  with   $d-k_i$ wires connected (if $k_i=0$, then due to the qudit Hopf law, there is no connection). %when $j$ has the d-nary expansion $j= \sum_{i=0}^{m-1}k_id^i, 0\leq k_i \leq d-1$. 
   The diagram in (\ref{wstyleqduditnformeq}) can be seen as a compressed version of the normal form as shown in  (\ref{normalformdit}), but is obviously easier  to be generalised to a normal form of a vector whose elements belong to arbitrary commutative semirings. Below we give a verification for the correctness of this new normal form.

By the interpretation of the W spider, we have 
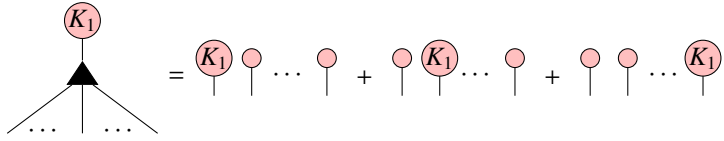
\[  %
	\beginpgfgraphicnamed{TikZit//k1onqditw}
	\InputIfFileExists{TikZit//k1onqditw.tikz}{}{\input{./figures/TikZit//k1onqditw.tikz}}%
	\endpgfgraphicnamed
 \]   
Considering  the general j-th term in the above sum decomposition,   we get
   \[  %
	\beginpgfgraphicnamed{TikZit//k1onqditwsumterms2v2}
	\InputIfFileExists{TikZit//k1onqditwsumterms2v2.tikz}{}{\input{./figures/TikZit//k1onqditwsumterms2v2.tikz}}%
	\endpgfgraphicnamed
 \]     
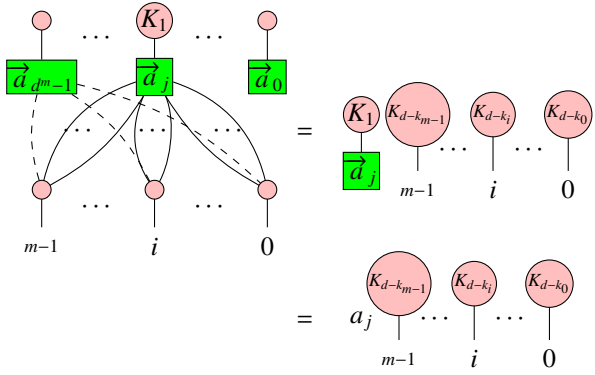
   \[=a_j\ket{ k_{m-1} \cdots k_{i} \cdots  k_{0}}  \] 
   By Lemma \ref{qbitstovectorditlm}, 
   \[ \ket{ k_{m-1} \cdots k_{i} \cdots  k_{0}}   =e_j
   \]
  Therefore the whole sum is $\sum_{j=0}^{d^{m-1}}a_je_j$ which is exactly the vector $$ \begin{pmatrix}
        a_0  \\  a_1\\
        \vdots \\ a_{d^m-2}\\
        a_{d^m-1} \end{pmatrix}.$$
  \section{Qufinite ZX-calculus}
 Based on the qudit ZX-calculi introduced in previous sections, now we set up a unified framework which we call qufinite ZX-calculus. The main idea is to label each wire with its dimension and add two new generators called dimension-splitter and dimension-binder respectively which were first presented in \cite{qwangslides}  and then deployed in  \cite{cwiconscious}. Note that  a wire labelled with 1 will be depicted as empty, as usually did.
 
 % Explain generators and rules. 
First we give the generators of   qufinite ZX-calculus.

\begin{table}[!h]
\begin{center} 
\begin{tabular}{|r@{~}r@{~}c@{~}c|r@{~}r@{~}c@{~}c|}
\hline
&& & %
	\beginpgfgraphicnamed{TikZit//generalgreenspiderqdit}
	\InputIfFileExists{TikZit//generalgreenspiderqdit.tikz}{}{\input{./figures/TikZit//generalgreenspiderqdit.tikz}}%
	\endpgfgraphicnamed
  & &&& %
	\beginpgfgraphicnamed{TikZit//HHdagd}
	\InputIfFileExists{TikZit//HHdagd.tikz}{}{\input{./figures/TikZit//HHdagd.tikz}}%
	\endpgfgraphicnamed
\\\hline
&& & %
	\beginpgfgraphicnamed{TikZit//triangled}
	\InputIfFileExists{TikZit//triangled.tikz}{}{\input{./figures/TikZit//triangled.tikz}}%
	\endpgfgraphicnamed
  & &&& %
	\beginpgfgraphicnamed{TikZit//triangledinv}
	\InputIfFileExists{TikZit//triangledinv.tikz}{}{\input{./figures/TikZit//triangledinv.tikz}}%
	\endpgfgraphicnamed
\\\hline
&& & %
	\beginpgfgraphicnamed{TikZit//idqudit}
	\InputIfFileExists{TikZit//idqudit.tikz}{}{\input{./figures/TikZit//idqudit.tikz}}%
	\endpgfgraphicnamed
  & &&& %
	\beginpgfgraphicnamed{TikZit//swapd}
	\InputIfFileExists{TikZit//swapd.tikz}{}{\input{./figures/TikZit//swapd.tikz}}%
	\endpgfgraphicnamed
\\\hline
&& & %
	\beginpgfgraphicnamed{TikZit//capdit}
	\InputIfFileExists{TikZit//capdit.tikz}{}{\input{./figures/TikZit//capdit.tikz}}%
	\endpgfgraphicnamed
  & &&& %
	\beginpgfgraphicnamed{TikZit//cupdit}
	\InputIfFileExists{TikZit//cupdit.tikz}{}{\input{./figures/TikZit//cupdit.tikz}}%
	\endpgfgraphicnamed
\\\hline
&& & %
	\beginpgfgraphicnamed{TikZit//binderdit}
	\InputIfFileExists{TikZit//binderdit.tikz}{}{\input{./figures/TikZit//binderdit.tikz}}%
	\endpgfgraphicnamed
  & &&& %
	\beginpgfgraphicnamed{TikZit//binderditflip}
	\InputIfFileExists{TikZit//binderditflip.tikz}{}{\input{./figures/TikZit//binderditflip.tikz}}%
	\endpgfgraphicnamed
\\\hline
\end{tabular} \caption{Generators of qufinite ZX-calculus, where  $d, m,n\in \mathbb N, d\geq 2; \protect\overrightarrow{\alpha_d}=(a_1,\cdots, a_{d-1}); a_i\in \mathbb C; i \in \{1,\cdots, d-1\}; j \in \{0, 1,\cdots, d-1\}; s, t \in \mathbb N \backslash\{0\}$. }\label{qbzxgeneratordit}
\end{center}
\end{table}

\FloatBarrier
\begin{remark}
  The two diagrams at the bottom of the table of generators are called  dimension-binder and  dimension-splitter respectively, in the qubit case they are similar to the divider and gatherer in the  as introduced in \cite{tdsscalarble}.
 
Since now wires are labelled with any positive integers, the category of  diagrams is not a PROP anymore,  but still a compact closed category.
\end{remark}

The rules of qufinite ZX-calculus can be divided in two parts: one part is the same as the qudit rules except each wire labelled with an integer $d$,  the other part has dimension-splitter and dimension-binder involved. Below we only give the rules of the second part which were partly shown in \cite{cwiconscious}.

 \begin{figure}[!h]
\begin{center} 
\[
\quad \qquad\begin{array}{|cc|}
\hline
	\beginpgfgraphicnamed{TikZit//binderunitary1}
	\InputIfFileExists{TikZit//binderunitary1.tikz}{}{\input{./figures/TikZit//binderunitary1.tikz}}%
	\endpgfgraphicnamed
&%
	\beginpgfgraphicnamed{TikZit//binderunitary2}
	\InputIfFileExists{TikZit//binderunitary2.tikz}{}{\input{./figures/TikZit//binderunitary2.tikz}}%
	\endpgfgraphicnamed
\\
    &\\ 
	\beginpgfgraphicnamed{TikZit//binderassoc}
	\InputIfFileExists{TikZit//binderassoc.tikz}{}{\input{./figures/TikZit//binderassoc.tikz}}%
	\endpgfgraphicnamed
&%
	\beginpgfgraphicnamed{TikZit//bindergspider}
	\InputIfFileExists{TikZit//bindergspider.tikz}{}{\input{./figures/TikZit//bindergspider.tikz}}%
	\endpgfgraphicnamed
\\
    &\\ 
	\beginpgfgraphicnamed{TikZit//binderwith1rt}
	\InputIfFileExists{TikZit//binderwith1rt.tikz}{}{\input{./figures/TikZit//binderwith1rt.tikz}}%
	\endpgfgraphicnamed
&%
	\beginpgfgraphicnamed{TikZit//binderwith1lt}
	\InputIfFileExists{TikZit//binderwith1lt.tikz}{}{\input{./figures/TikZit//binderwith1lt.tikz}}%
	\endpgfgraphicnamed
\\
    &\\ 
  		  		\hline  
  		\end{array}\]      
  	\end{center}
  	\caption{Qufinite ZX-calculus rules  involving dimension-splitter and dimension-binder,  where $\protect\overrightarrow{1}_d=\protect\overbrace{(1,\cdots,1)}^{d-1}, \protect\overrightarrow{0}_d=\protect\overbrace{(0,\cdots,0)}^{d-1},  \protect\overrightarrow{\alpha_d}=(a_1,\cdots, a_{d-1}), \protect\overrightarrow{\beta_d}=(b_1,\cdots, b_{d-1}), a_k, b_k\in  \mathbb C, k \in \{1,\cdots, d-1\},  j \in \{ 1,\cdots, d-1\}, s, t, u \in \mathbb N \backslash\{0\}.$}\label{qufiniterules2}
  \end{figure}  
  
 \FloatBarrier

%\[
%\left\llbracket \tikzfig{TikZit//generalgreenspiderqdit} \right\rrbracket=\sum_{i=0}^{d-1}a_j\ket{i}^{\otimes m}\bra{i}^{\otimes n}, a_0=1, a_i\in \mathcal{S},
%\]
%\[
%\left\llbracket \tikzfig{TikZit//redspider0pd} \right\rrbracket=\sum_{\substack{0\leq i_1, \cdots, i_m,  j_1, \cdots, j_n\leq d-1\\ i_1+\cdots+ i_m\equiv  j_1+\cdots +j_n(mod~ d)}}\ket{i_1, \cdots, i_m}\bra{j_1, \cdots, j_n},
%\]
%\[
%\left\llbracket\tikzfig{TikZit//redclassicd}\right\rrbracket= \sum_{i=0}^{d-1}\ket{i}\bra{i\oplus j},\quad
  %\left\llbracket\tikzfig{TikZit//triangled}\right\rrbracket=\ket{0}\bra{0}+\sum_{i=1}^{d-1}(\ket{0}+\ket{i})\bra{i}, \quad
%\left\llbracket\tikzfig{TikZit//idqudit}\right\rrbracket= \sum_{i=0}^{d-1}\ket{i}\bra{i}, 
   %\]
\[
 \left\llbracket%
	\beginpgfgraphicnamed{TikZit//binderdit}
	\InputIfFileExists{TikZit//binderdit.tikz}{}{\input{./figures/TikZit//binderdit.tikz}}%
	\endpgfgraphicnamed
\right\rrbracket= \sum_{k=0}^{s-1}\sum_{l=0}^{t-1}\ket{kt+l}\bra{kl}, 
 \quad \quad
  \left\llbracket%
	\beginpgfgraphicnamed{TikZit//binderditflip}
	\InputIfFileExists{TikZit//binderditflip.tikz}{}{\input{./figures/TikZit//binderditflip.tikz}}%
	\endpgfgraphicnamed
\right\rrbracket= \sum_{k=0}^{st-1}\ket{[\frac{k}{t}]}\ket{k-t[\frac{k}{t}]}\bra{k}, \quad\quad   \left\llbracket%
	\beginpgfgraphicnamed{TikZit//swapd}
	\InputIfFileExists{TikZit//swapd.tikz}{}{\input{./figures/TikZit//swapd.tikz}}%
	\endpgfgraphicnamed
\right\rrbracket=\sum_{k=0}^{s-1}\sum_{l=0}^{t-1}\ket{kl}\bra{lk},
 \]

%\[
% \left\llbracket\tikzfig{TikZit//swapd}\right\rrbracket= \sum_{k=0}^{s-1}\sum_{l=0}^{t-1}\ket{kl}\bra{lk}, 
 %\quad
%  \left\llbracket\tikzfig{TikZit//capdit}\right\rrbracket= \sum_{i=0}^{s-1}\ket{i}\ket{i},  \quad
 %  \left\llbracket\tikzfig{TikZit//cupdit}\right\rrbracket= \sum_{i=0}^{s-1}\bra{i}\bra{i},
   %   \]

\[  \llbracket D_1\otimes D_2  \rrbracket =  \llbracket D_1  \rrbracket \otimes  \llbracket  D_2  \rrbracket, \quad 
 \llbracket D_1\circ D_2  \rrbracket =  \llbracket D_1  \rrbracket \circ  \llbracket  D_2  \rrbracket,
  \]
where 
$s, t \in \mathbb N \backslash\{0\},   \bra{i} =\overbrace{(\underbrace{0,\cdots,1}_{i+1}, \cdots, 0)}^{d}, ~ \ket{i}=(\overbrace{(\underbrace{0,\cdots,1}_{i+1}, \cdots, 0)}^{d})^T,  i \in \{0, 1,\cdots, d-1\},$  and $[r]$ is the integer part of a real number $r$.

\begin{remark}
In the 1-dimensional Hilbert space $H_1=\mathbb C$, we make the convention that $\ket{0}=1$.
\end{remark}

\subsection{Normal form and universality}
Given an arbitrary $s\times t$ matrix $M$ over $\mathbb C$: 
\begin{equation}\label{anymatrix}
M=\begin{blockarray}{cccc}
\begin{block}{(cccc)}
     a_0 & a_1  &\cdots & a_{t-1} \\
      a_t & a_{t+1}  &\cdots & a_{t+t-1} \\
     \vdots    & \ddots & &  \vdots  \\
       a_{kt }& a_{kt+1}  &\cdots & a_{kt+t-1} \\
       \vdots    & &   \ddots&  \vdots   \\
       a_{(s-1)t }& a_{(s-1)t+1} &\cdots & a_{(s-1)t+t-1} \\
\end{block}
\end{blockarray}=(a_{kt+l})_{0\leq k \leq s-1, 0\leq l \leq t-1}
\end{equation}

\begin{theorem}
The matrix shown in (\ref{anymatrix}) can be represented by the following diagram:
 \begin{equation}\label{nomalformfnit}
	\beginpgfgraphicnamed{TikZit//normalformfit2}
	\InputIfFileExists{TikZit//normalformfit2.tikz}{}{\input{./figures/TikZit//normalformfit2.tikz}}%
	\endpgfgraphicnamed

\end{equation}
where $1\leq k \leq st-1$,  $ \overrightarrow{a}_i=\overbrace{(0,\cdots,0, a_i)}^{st-1}, 0\leq i \leq st-2, \overrightarrow{a}_{st-1}=\overbrace{(1,\cdots,1, a_{st-1})}^{st-1}$. %$a_i, a_{st-1} \in  \mathbb C$.
\end{theorem}

\begin{proof}
By the normal form for qudits, the diagram excluding the dimension-splitter represents the following vector 
\[
 \begin{pmatrix}
        a_0  \\  a_1\\
        \vdots \\ a_{st-2}\\
        a_{st-1} \end{pmatrix}=\sum_{i=0}^{st-1}a_i\ket{i}
\]
Then the whole diagram (\ref{nomalformfnit}) represents 
\[(I_s\otimes  \sum_{i=0}^{t-1}\bra{i}\bra{i})( \sum_{k=0}^{st-1}\ket{[\frac{k}{t}]}\ket{k-t[\frac{k}{t}]}\bra{k}\otimes I_t)(\sum_{i=0}^{st-1}a_i\ket{i}\otimes I_t)\]

Then in matrix form, its element in $v$-th row and the $l$-th column ($0\leq v\leq s-1, 0\leq l\leq t-1$) is 
\[
\begin{array}{l}
\bra{v}(I_s\otimes  \sum_{i=0}^{t-1}\bra{i}\bra{i})( (\sum_{k=0}^{st-1}\ket{[\frac{k}{t}]}\ket{k-t[\frac{k}{t}]}\bra{k})\otimes I_t)((\sum_{i=0}^{st-1}a_i\ket{i})\otimes I_t)\ket{l}\vspace{0.5cm}\\
=\bra{v}(I_s\otimes  \sum_{i=0}^{t-1}\bra{i}\bra{i})( \sum_{k=0}^{st-1}a_k\ket{[\frac{k}{t}]}\ket{k-t[\frac{k}{t}]}\ket{l}\vspace{0.5cm}\\
=\bra{v}\sum_{k-t[\frac{k}{t}]=l}a_k\ket{[\frac{k}{t}]}\vspace{0.5cm}\\
=\sum_{k-t[\frac{k}{t}]=l, v=[\frac{k}{t}]}a_k\bra{v}\ket{[\frac{k}{t}]}=a_{l+vt}

\end{array}
\]

\end{proof}

\begin{remark}
Given matrix (\ref{anymatrix}), the diagram of form (\ref{nomalformfnit}) is unique, thus will be called a norm form of qufinite ZX-calculus.  This representation actually works for matrices over arbitrary commutative semirings as well.

Another interesting observation is that the requirement of $t=u$ when defining multiplication of matrices $M_{s\times t}N_{u\times v}$  can now be clearly seen as the need of type matching  for composing the normal form  (\ref{nomalformfnit}).

%To simplify normal form such that we need not know all the elements of  the corresponding matrix, as fewer as possible.
\end{remark}

% \fi
\section{Conclusion  and further work}
 In this paper, we generalise the qubit ZX-calculus to qudit ZX-calculus in any finite dimension by introducing Z spider with complex-number phase vector and generalised triangle node as new generators. As a consequence, we obtained qudit rewriting rules  which can be seen as direct generalisation of qubit rules. Furthermore, we construct a normal form for any qudit vectors, which exhibits universality for qudit ZX-calculus. Finally, we propose qufinite ZX-calculus as a unified framework for qudit ZX-calculi in all finite dimensions, with a normal form for matrix of any finite size.    
 
The next work would naturally be to prove the completeness of qudit ZX-calculus and qufinite ZX-calculus, over complex numbers and arbitrary commutative semirings respectively, following the method used in \cite{qwangnormalformbit,qwangrsmring}. Another interesting work would be to give a fine-grained version of  the diagrammatic reconstruction of finite quantum theory \cite{Selby2021reconstructing} within the framework of qufinite ZX-calculus.

%Another interesting work would be considering reconstruction of finite quantum theory in the framework of qufinite ZX-calculus which focuses on compositionality,  without resort to probability theory or sum structures.

  \section*{Acknowledgements} 

%The author would like to thank Bob Coecke, KangFeng  Ng, Aleks Kissinger for useful discussions. 
The author would like to acknowledge the grant FQXi-RFP-CPW-2018. The author also thanks useful discussions at the ZX-calculus seminar. 

\bibliographystyle{eptcs}
\bibliography{generic} 

\end{document}